\documentclass[11pt,a4paper]{article}
\usepackage[a4paper,margin=1.02in,bottom=1.04in,includefoot]{geometry}
\usepackage[T1]{fontenc}
\usepackage{lmodern}
\usepackage{iftex}
\ifPDFTeX
  \input{glyphtounicode}
\fi
\usepackage{graphicx}
\usepackage{amsmath,amssymb,amsfonts,amsthm,mathrsfs,mathtools}
\usepackage[title]{appendix}
\usepackage{xcolor,textcomp,booktabs}
\usepackage[numbers,sort&compress]{natbib}
\usepackage[hidelinks]{hyperref}

\renewcommand{\le}{\leqslant}

\renewcommand{\ge}{\geqslant}

\newcommand{\cC}{\mathcal{C}}

\newcommand{\cL}{\mathcal{L}}

\newcommand{\cS}{\mathcal{S}}

\newtheoremstyle{arxivplain}{\topsep}{\topsep}{\itshape}{}{\bfseries}{.}{\newline}{}
\theoremstyle{arxivplain}
\newtheorem{theorem}{Theorem}
\newtheorem{proposition}[theorem]{Proposition}
\newtheorem{corollary}[theorem]{Corollary}
\newtheorem{lemma}[theorem]{Lemma}

\newtheorem{construction}{Construction}
\newtheorem{conjecture}[theorem]{Conjecture}
\theoremstyle{remark}
\newtheorem{example}{Example}
\newtheorem{remark}{Remark}
\theoremstyle{definition}

\hypersetup{
  bookmarksdepth=subsubsection,
  pdftitle={Multiset Deletion Codes: Cyclic Syndrome Constructions, Bounds, and Exact Results},
  pdfauthor={Avraham Kreindel, Isaac Barouch Essayag, Aryeh Lev Zabokritskiy (Yohananov)},
  pdflang={en},
  pdfkeywords={Multiset codes, deletion correction, Sidon sets, constant-weight codes (Manhattan distance), discrete simplex}
}
\title{Multiset Deletion Codes: Cyclic Syndrome Constructions, Bounds, and Exact Results}
\author{\begin{minipage}{\dimexpr\textwidth-2\tabcolsep\relax}\centering
  Avraham Kreindel\textsuperscript{1}, Isaac Barouch Essayag\textsuperscript{2}\\[0.3em]
  Aryeh~Lev Zabokritskiy (Yohananov)\textsuperscript{3,*}\\[0.8em]
  {\small
  \textsuperscript{1}Department of Computer Science, Reichman University, Herzliya, Israel\\[0.3em]
  \textsuperscript{2}Research Assistant, MIGAL -- Galilee Research Institute/Tel-Hai University of Kiryat Shmona and the Galilee, Kiryat Shmona, Israel\\[0.3em]
  \textsuperscript{3}Department of Computer Science, MIGAL -- Galilee Research Institute/Tel-Hai University of Kiryat Shmona and the Galilee, Kiryat Shmona, Israel\\[0.6em]
  \textsuperscript{1}\href{mailto:avrahamkreindel@gmail.com}{\texttt{avrahamkreindel@gmail.com}}; \textsuperscript{2}\href{mailto:isaac.es@migal.org.il}{\texttt{isaac.es@migal.org.il}}\\[0.3em]
  \textsuperscript{3,*}Corresponding author: \href{mailto:yuhanalev@telhai.ac.il}{\texttt{yuhanalev@telhai.ac.il}}
  }
\end{minipage}}
\date{}
\begin{document}
\maketitle
\begin{abstract}
We determine sharp class-size estimates for multiset codes over \(q\) symbols that correct up to \(t\) deletions. For fixed \(q,t\ge2\), we prove that the syndrome classes defined by Nathanson's geometric-sum labels modulo \(M=1+t+\cdots+t^{q-1}\) differ from their average size by \(O(n^{q/p-1})\), where \(n\) is the blocklength and \(p\) is the least prime divisor of \(q\). This exponent is sharp, and every prescribed class has size asymptotic to the average. For prime \(q\), all classes have equal size when \(n\equiv-1,\ldots,-(q-1)\pmod M\). We also present our radix construction, whose labels are zero and the first \(q-1\) powers of \(t+1\), starting at \(1\), modulo \(t(t+1)^{q-2}+1\). We determine its sharp class-size deviation exponent. For even \(q\ge4\), its syndrome classes have a smaller deviation exponent, although their average size is smaller. Over four symbols and every \(t\ge2\), we show that the radix labels still correct up to \(t\) deletions modulo a reduced modulus, which gives bounded deviations and eventually larger classes than either unmodified construction. For a fixed number \(k\ge2\) of surviving symbols, we prove that some maximum code contains all constant words, reducing the remaining optimization to multisets with multiplicities below \(k\). Finally, for one deletion at length four over \(q\ge4\) symbols, we show that the best bound from fractional weights on deletion outputs, rounded down, equals the maximum code size exactly when \(q\equiv2,4\pmod6\).
\end{abstract}
\medskip
\noindent\textbf{Keywords.} Multiset codes, deletion correction, Sidon sets, constant-weight codes (Manhattan distance), discrete simplex
\par\smallskip
\noindent\textbf{Mathematics Subject Classification (2020).} 94B25, 94B05, 94B60, 94B65
\renewcommand{\labelitemi}{\(\bullet\)}
\section{Introduction}
Multiset deletion codes describe communication in which the receiver observes symbol multiplicities but not their order~\cite{KovacevicVukobratovic2013,KovacevicTan2018}. A codeword is a multiset of size \(n\) over \(q\) symbols. A code corrects up to \(t\) deletions if every remaining multiset determines the transmitted codeword. We study explicit constructions and the maximum possible code size.

In a syndrome construction, each symbol has a label in a finite Abelian group, and the sum of the labels, counted with multiplicity, is the syndrome. Suitable labels partition the multisets into deletion-correcting syndrome classes~\cite{KovacevicTan2018}; it suffices that every sum of \(t\) labels, repetitions allowed, determines its summands. Existing constructions use lattices~\cite{KovacevicTan2018,KovacevicCross2022} and finite fields~\cite{BoseChowla1960,Varshamov1973}, including our polynomial constructions~\cite{KreindelEssayagZabokritskiy2026}. Nathanson~\cite[Corollary~2.2]{Nathanson2025} provided the geometric-sum labels \(0,1,1+t,\ldots,1+t+\cdots+t^{q-2}\) with the required distinct-sum property in the integers. An averaging argument guarantees a large class, but does not give the size of a class whose syndrome is prescribed in advance. Bounds for every class give an asymptotic size guarantee for a prescribed syndrome, without searching for a largest class.

Our first main result gives sharp class-size estimates for the construction based on Nathanson's labels. For fixed \(q,t\ge2\) and modulus \(M=1+t+\cdots+t^{q-1}\), every class differs from the average by \(O(n^{q/p-1})\), where \(p\) is the least prime divisor of \(q\). We prove that this exponent is best possible. For prime \(q\), the deviations are periodic in the blocklength, and all classes have equal size when \(n\equiv-1,\ldots,-(q-1)\pmod M\).

We also present our radix construction, whose labels are zero and the first \(q-1\) powers of \(t+1\), starting at \(1\), modulo \(t(t+1)^{q-2}+1\). We determine its sharp class-size deviation exponent. For \(q\ge4\) its classes are smaller on average than the geometric classes, but for every even \(q\ge4\) and \(t\ge2\) we prove that their deviations have a smaller exponent; in particular, the deviations are bounded for \(q=4,6\). Over four symbols and every \(t\ge2\), we further reduce the radix modulus below both original moduli while preserving deletion correction and bounded deviations; at \(t=2\) this recovers the planar difference set \(\{0,1,3,9\}\) modulo \(13\) used in~\cite[Example~2]{KovacevicTan2018}. At all sufficiently large blocklengths, every class of this reduced-modulus construction is larger than every class of either original construction.

For each fixed \(q\ge6\), we obtain another cyclic syndrome construction with bounded deviations and a modulus asymptotic to \(t^{q-1}/3\) as \(t\to\infty\), one third of the geometric modulus. For this construction, we prove the required distinct-sum property of a sixth label proposed by O'Bryant~\cite{OBryant2025}, building on the five-label result of Nathanson and O'Bryant~\cite{NathansonOBryant2024}. For fixed \(q\), the geometric, radix and six-label moduli all have the optimal order \(t^{q-1}\) among syndrome constructions based on finite Abelian groups~\cite{KovacevicTan2017}. The comparison with finite-field constructions depends on the parameters: Varshamov's power-sum construction~\cite{Varshamov1973} uses fewer syndromes when \(t\ge2\) is fixed and \(q\) grows, whereas for \(q\ge3\) and \(t\ge q-1\) we prove that the geometric construction uses fewer syndromes than the power-sum construction over any finite field of characteristic above \(t\). We also give direct decoders for the geometric and radix codes.

We next study the maximum size of a multiset code of length \(n\) over \(q\) symbols when \(k=n-t\ge2\) symbols survive. We prove that some maximum code contains all \(q\) constant words, and that its remaining words form a maximum code among multisets with every multiplicity below \(k\). This extends the two-survivor results at lengths three and four of Chen, Ma and Zhang~\cite{ChenMaZhang2021}. For three survivors, combining this reduction with their constant-weight code results and those of Wei and Zhang~\cite{WeiZhang2023} gives exact code sizes for \(n\ge q+1\) and the leading term for fixed \(n\) as \(q\) grows.

For a single deletion, fractional covering gives upper bounds by weighting deletion outputs~\cite{FazeliVardyYaakobi2015,CullinaKiyavash2016}. At length four, Chen, Ma and Zhang~\cite{ChenMaZhang2021} express the exact code size through quadruple packings. We determine the optimal fractional-cover bound at length four for every \(q\ge4\). Its value is \(q\) plus the unrounded value of Wei and Zhang's bound~\cite{WeiZhang2023} on the codewords with every multiplicity below three. We prove that rounding this fractional bound down gives the maximum code size exactly when \(q\equiv2,4\pmod6\), the orders for which Steiner quadruple systems exist~\cite{Hanani1960}.

The class counts for the single-deletion labels \(0,\ldots,q-1\) in \(\mathbb Z_q\) are known~\cite{Panyushev2011,MuratovicWang2015}. For prime \(q\), we conjecture that the largest class is optimal at all sufficiently large blocklengths. Together with the elementary two-symbol case, the optima for three symbols of Geramita, Gregory and Roberts~\cite{GeramitaGregoryRoberts1986} and for five and seven symbols from the third author's work~\cite{ZabokritskiyWeightedLocalCovers2026} establish this conjecture for \(q=2,3,5,7\). For four symbols, we prove that the zero-syndrome class over the additive group of the field of order four attains the optimum of Geramita, Gregory and Roberts at every blocklength. At every positive even length \(n\), it exceeds the largest class over \(\mathbb Z_4\) by \(\lceil n/4\rceil\) words; at odd lengths both groups attain the optimum.

Section~\ref{sec:pre} describes the distance and syndrome formulations. Section~\ref{sec:cyclic} proves the balance theorems, gives the reduced-modulus constructions, and compares the codes with lattice and finite-field constructions. Section~\ref{sec:bounds} compares upper bounds and treats codes with few surviving symbols. Section~\ref{sec:elementary-benchmarks} studies single deletions and finite-length optimality.

\section{Multisets, distance and syndromes}\label{sec:pre}
Fix an alphabet \(\Sigma=\{0,\ldots,q-1\}\), \(q\ge2\). We identify a multiset \(S\) of cardinality \(n\), also called a word of length \(n\), with its multiplicity vector \(\mathbf x=(x_0,\ldots,x_{q-1})\), where \(x_i=S(i)\) is the multiplicity of \(i\). The set of such multiplicity vectors is the discrete simplex
\[
 \cS_{n,q}=\{\mathbf x\in\mathbb Z_{\ge0}^q:\textstyle\sum_i x_i=n\},
 \qquad V_{n,q}=|\cS_{n,q}|=\binom{n+q-1}{q-1}.
\]
Multiset intersection takes the smaller multiplicity in each coordinate. If \(S,T\in\cS_{n,q}\) have multiplicity vectors \(\mathbf x,\mathbf y\), their deletion distance is
\[
 d(S,T)=n-|S\cap T|=\tfrac12\|\mathbf x-\mathbf y\|_1.
\]
The second equality follows by summing \(|x_i-y_i|=x_i+y_i-2\min(x_i,y_i)\). The largest possible common output of \(S\) and \(T\) is their intersection. Thus a code \(\cC\subseteq\cS_{n,q}\) corrects up to \(t\) deletions exactly when each distinct pair has distance at least \(t+1\). Write \(S_q(n,t)\) for the maximum code size. A nonempty code has redundancy \(R(\cC)=\log_q(V_{n,q}/|\cC|)\).

For labels \(b_0,\ldots,b_{q-1}\) in a finite Abelian group \(G\), the syndrome of \(\mathbf x\) is \(\sum_i x_i b_i\in G\), and the syndrome class of \(a\in G\) is the set of multiplicity vectors in \(\cS_{n,q}\) with syndrome \(a\). Each nonempty class is the intersection of the simplex with a translate of the kernel lattice \(\{\mathbf z\in A^{q-1}:\sum_i z_i b_i=0\}\), where \(A^{q-1}=\{\mathbf z\in\mathbb Z^q:\sum_i z_i=0\}\); this is a linear multiset code in the sense of Kova\v{c}evi\'c and Tan~\cite[Definition~2]{KovacevicTan2018}. For example, at \(q=3,n=4\) the labels \(1,3,0\) in \(\mathbb Z_7\) (the radix labels of Subsection~\ref{subsec:radix} with \(t=2\)) and the syndrome \(3\) give the condition \(x_0+3x_1\equiv3\pmod 7\), which picks out the multiplicity vectors \((0,1,3),(1,3,0),(3,0,1)\), at pairwise distance three. The group need not be cyclic: in the same simplex the translate \((4,0,0)+4A^2\) of the kernel lattice of the labels \((0,0),(1,0),(0,1)\) in the noncyclic group \(\mathbb Z_4^2\) picks out just the three constant words, at pairwise distance four.

A set of labels is a \(B_t\)-set if all sums of exactly \(t\) labels, allowing repetitions, determine the summands up to order. This implies uniqueness for each fixed number \(r\le t\) of summands: append \(t-r\) copies of the same label to both sums. Decoding needs uniqueness only among sums with equally many summands, because the output length reveals \(r\).

\section{Class sizes for two cyclic syndrome constructions}
\label{sec:cyclic}

We first determine how far the class sizes for Nathanson's geometric labels
can deviate from their average, then give the radix construction and its sharp class-size estimates.
The final subsection compares these with other constructions, separating the size of the syndrome group from
the balance of its classes about their own average; it also constructs labels with a smaller prime modulus for six or more symbols (Proposition~\ref{prop:six-label-prime}).

\subsection{Geometric labels and their balance}
\label{subsec:geometric-balance}

For \(t\ge2\), consider the first \(q\) terms of the sequence
\[
 b_i=\frac{t^i-1}{t-1}\quad(i\ge0),\qquad M=b_q=1+t+\cdots+t^{q-1}.
\]
Nathanson~\cite[Corollary~2.2]{Nathanson2025} proved that the infinite
sequence is a \(B_t\)-set in the integers. Its first \(q\) terms remain a
\(B_t\)-set modulo \(M\): every sum of \(t\) such terms lies between \(0\)
and \(tb_{q-1}=M-1\), so reduction identifies no two distinct sums.
We also give a direct decoder. If \(E\) is a deletion pattern of known
size \(r\le t\), the syndrome difference determines its integer sum
\(\sigma=\sum_i E(i)b_i\) in this same interval. Then
\[
 (t-1)\sigma+r=\sum_{i=0}^{q-1}E(i)t^i.
\]
The base-\(t\) digits on the right are the multiplicities unless some
\(E(j)=t\). In that case \(r=t\), all deletions have symbol \(j\), and the
integer is \(t^{j+1}\). This case is recognizable: its digit sum is \(1\),
whereas an expansion without a carry has digit sum \(r=t\). Hence recovery
is unique and requires no table of deletion patterns. Consequently each syndrome class modulo \(M\) corrects up to \(t\) deletions.

To count the words in a prescribed class, write \(D_a(n)\) for the number of
length-\(n\) words with \(\sum_i b_i x_i=a\bmod M\).
Muratovi\'c-Ribi\'c and Wang~\cite[Section~3]{MuratovicWang2015} give a
counting formula for multisets drawn from any subset of a finite Abelian
group; for the deviations we argue directly with characters. For fixed \(a\), the deviation \(D_a(n)-V_{n,q}/M\) is a quasipolynomial in \(n\), and the degree of a quasipolynomial is determined by the largest pole order
of its generating function~\cite[Lemma~9]{BeckSamWoods2008}. A character that takes the same value on several labels produces a repeated pole, so the
arithmetic question is how often a character takes the same value on
different labels. The map \(x\mapsto tx+1\) cycles through the chosen
labels, and this cyclic symmetry determines that number.

\begin{theorem}[Balance for the geometric labels \(b_i\)]\label{thm:orbit-balance}
Fix \(q,t\ge2\) and let \(p\) be the smallest prime divisor of \(q\).
Put \(u=q/p\). Uniformly in \(a\),
\[
 D_a(n)=V_{n,q}/M+O_{q,t}(n^{u-1}).
\]
The exponent is attained: for at least one residue \(a\) the absolute
deviation \(|D_a(n)-V_{n,q}/M|\) is at least a positive constant times \(n^{u-1}\) for infinitely
many \(n\). If \(q\) is prime, the deviations are periodic with period dividing
\(M\); moreover, all classes are equal for
\(n\equiv-1,-2,\ldots,-(q-1)\pmod M\).
\end{theorem}

\begin{proof}
Character orthogonality gives
\[
 D_a(n)=\frac1M\sum_{\chi}\overline{\chi(a)}[z^n]
        \prod_{i=0}^{q-1}(1-z\chi(b_i))^{-1},
\]
where \(\chi\) runs over all characters of \(\mathbb Z_M\): expansion of
the product counts every multiplicity vector with its character weight, and the sum
over characters keeps precisely the multiplicity vectors of syndrome \(a\).
The trivial character contributes \(V_{n,q}/M\).
Consider a nontrivial character of order \(d\).
It gives the same value on two labels exactly when they agree modulo \(d\).
Since \(d\mid M\) and \(M\equiv1\pmod t\), the map \(x\mapsto tx+1\)
permutes \(\mathbb Z_d\). Starting at \(b_0=0\), it successively visits
\(b_1,\ldots,b_q\), with \(b_q=M\equiv0\pmod d\). Its orbit length \(k\)
therefore divides \(q\), and \(k>1\) since \(b_1=1\not\equiv0\pmod d\).
In the list of \(q\) labels each of the \(k\) character values occurs exactly
\(q/k\le u\) times. Every pole of its product therefore has order at most
\(u\); a partial-fraction term \(c(1-\lambda z)^{-j}\) contributes
\(c\lambda^n\binom{n+j-1}{j-1}\), with \(|\lambda|=1\) and \(j\le u\).
The finitely many nontrivial characters together contribute
\(O_{q,t}(n^{u-1})\), uniformly in \(a\).

For attainment take \(d=b_p=1+t+\cdots+t^{p-1}>1\), a divisor of \(M\)
because \(p\mid q\). The orbit returns at step \(p\), so its length is
exactly \(p\). Choose a character \(\chi\) of order \(d\). Its generating
function \(\prod_i(1-z\chi(b_i))^{-1}\) has a pole of order exactly \(u\) at \(1\): its numerator is \(1\), and exactly the \(u\) labels \(b_0,b_p,\ldots,b_{q-p}\) are \(\equiv0\pmod d\).
Partial fractions at the \(M\)-th roots of unity show that on each residue
class of \(n\bmod M\) its coefficients are polynomials in \(n\) of degree
at most \(u-1\). At least one has degree \(u-1\): otherwise the generating function, reassembled from these
polynomial sequences over the finitely many residue classes, would have a pole
of order less than \(u\) at \(1\) (for \(u=1\) all sequences would vanish).
These coefficients equal \(\sum_a\chi(a)(D_a(n)-V_{n,q}/M)\), since
\(\sum_a\chi(a)=0\). On a residue class of \(n\) where the degree is \(u-1\), the triangle inequality gives, for each large \(n\) in that class, some \(a\) with \(|D_a(n)-V_{n,q}/M|\) at least a constant times \(n^{u-1}\); since there are only \(M\) residues \(a\), one fixed \(a\) does so for infinitely many \(n\).

If \(q\) is prime, every nontrivial character gives \(q\) distinct values.
Its generating function can thus be written \(P(z)/(1-z^M)\), with
\(\deg P=M-q\). The coefficients are periodic and vanish on the last
\(q-1\) residue classes modulo \(M\). Only the trivial character then
contributes to the count, proving exact equality there.
\end{proof}

\subsection{The radix construction and its class sizes}
\label{subsec:radix}

Our radix construction uses the powers \(1,t+1,\ldots,(t+1)^{q-2}\) and \(0\) as labels in the
construction of Kova\v{c}evi\'c and Tan based on \(B_t\)-sets~\cite[Section~IV-A]{KovacevicTan2018}.
The sum of the deleted labels can be read as an integer in base \(t+1\).
We first recover the deleted multiplicities and then determine the sharp
uniform exponent for class-size deviations from the average.

\begin{construction}[Radix labels]
\label{const:cyclic}
Let \(q\ge2\) and \(1\le t\le n\). Set \(b=t+1\) and \(m=t b^{q-2}+1\).
Assign the labels
\[
   f(i)=b^i\quad(0\le i\le q-2),\qquad f(q-1)=0
\]
in \(\mathbb Z_m\). For a multiplicity vector \(\mathbf x\), write
\(L(\mathbf x)=\sum_{i=0}^{q-2}b^i x_i\bmod m\).
For each \(a\in\mathbb Z_m\), take the length-\(n\) code
\[
   \cC_q(a)=\{\mathbf x\in\cS_{n,q}:L(\mathbf x)=a\}.
\]
\end{construction}

As in Section~\ref{sec:pre}, every nonempty class is a linear multiset code.

\begin{lemma}[The deleted digits]
\label{lem:injective-syndrome}
For every \(0\le r\le t\), a deleted multiset of known size \(r\) is
determined by the sum of its labels in \(\mathbb Z_m\).
\end{lemma}

\begin{proof}
If \(e_i\) copies of symbol \(i\) were deleted, then
\[
   0\le \sum_{i=0}^{q-2} b^i e_i
      \le r b^{q-2}\le m-1.
\]
Thus reducing the sum modulo \(m\) loses no information. Since
\(0\le e_i\le t=b-1\), its base-\(b\) digits are
\(e_0,\ldots,e_{q-2}\), least significant first and padded with leading
zeros. The last count is \(e_{q-1}=r-\sum_{i=0}^{q-2}e_i\).
\end{proof}

Taking \(r=t\) gives the ordinary \(B_t\)-property. Knowing \(r\) matters:
adding a zero label preserves the syndrome but changes the number of deletions.

\begin{theorem}[Correction and size]
\label{thm:cyclic-correctness}
Every \(\cC_q(a)\) corrects up to \(t\) deletions. For some \(a^*\in\mathbb Z_m\),
\[
  |\cC_q(a^*)|\ge \frac1m\binom{n+q-1}{q-1},\qquad
  R(\cC_q(a^*))\le \log_q m .
\]
\end{theorem}

\begin{proof}
Given the received multiplicities \(\mathbf y\), compute the number of
deletions \(r=n-\sum_i y_i\) and the least nonnegative residue
\(s=(a-L(\mathbf y))\bmod m\). For a valid transmission, \(s\) is the
deleted-label sum in Lemma~\ref{lem:injective-syndrome}. Recover
\(\mathbf e\) from its digits and \(r\), and return \(\mathbf y+\mathbf e\).
For the size bound, the \(m\) classes partition the
\(\binom{n+q-1}{q-1}\) multiplicity vectors.
\end{proof}

\begin{proposition}[Direct decoding]
\label{prop:cyclic-complexity}
With \(a\) fixed, decoding and membership testing use \(O(q)\) arithmetic
operations on a multiplicity-vector input, or \(O(n+q)\) operations on a
list of symbols. No table of deletion patterns is needed.
\end{proposition}

\begin{proof}
Evaluate \(L(\mathbf y)\) by Horner's rule modulo \(m\). Recover each of the
first \(q-1\) counts by taking the remainder of \(s\) on division by \(b\),
then replacing \(s\) by the quotient. Reject if \(r\notin[0,t]\) or if
the recovered digits sum to more than \(r\); otherwise the missing
zero-label count is nonnegative and the reconstructed word has syndrome
\(a\). Testing membership requires only the length and syndrome checks.
A symbol list can first be counted in one pass, and the corrected list,
if requested, can be written in \(O(n)\) further operations.
\end{proof}

These are arithmetic-operation counts, not bit bounds uniform in \(q,t,n\);
the calculation uses counts up to \(n\) and arithmetic modulo \(m\).
For fixed \(q,t\), decoding list input is linear in \(n\) in the usual
word-RAM model. Selecting a largest class (Remark~\ref{rem:choice-of-residue})
and encoding messages by rank within it are separate tasks.

\paragraph{Size of a prescribed syndrome class.}

Averaging gives a finite-length guarantee for at least one residue, not
necessarily for a prescribed one. To estimate all classes at once, we count
repeated poles in their common generating function. The parameter \(\ell\)
below is the least index gap at which two of the powers \(1,b,\ldots,b^{q-2}\)
give the same pole for some nontrivial character, with \(\ell=q-1\) if no two do.

\begin{proposition}[Balance of the radix syndromes]
\label{prop:cyclic-balance-sharp}\label{prop:cyclic-balance}
Fix \(q\ge2\) and \(t\ge1\). Put
\[
 \ell=\min\bigl(\{1\le k\le q-2:\gcd(m,b^k-1)>1\}\cup\{q-1\}\bigr),
 \qquad E=\left\lfloor\frac{q-2}{\ell}\right\rfloor.
\]
Uniformly for \(a\in\mathbb Z_m\), as \(n\to\infty\),
\[
   |\cC_q(a)|=\frac1m\binom{n+q-1}{q-1}+O_{q,t}(n^E),
   \qquad R(\cC_q(a))=\log_q m+O_{q,t}(n^{E-q+1}).
\]
For some fixed \(a\) and \(c>0\), the absolute deviation of \(|\cC_q(a)|\)
from the average is at least \(cn^E\) for infinitely many \(n\). Moreover, \(E\le h-1\),
where \(h=\lceil(q-1)/2\rceil\) for odd \(q\), and
\(h=\lceil(q-1)/3\rceil\) for even \(q\).
The deviations are bounded for \(q=2,3,4,6\).
Whenever \(E=0\), every deviation is periodic in \(n\), with
period dividing \(m\).
For \(q=4\) every class has size \(V_{n,4}/m\) whenever
\(n\equiv-1,-2,-3\pmod m\).
\end{proposition}

\begin{proof}
Let \(\omega=\exp(2\pi i/m)\). Filtering the congruence \(L(\mathbf x)=a\)
and summing each coordinate over its nonnegative values gives
\[
   |\cC_q(a)|=\frac1m\sum_{j=0}^{m-1}\omega^{-aj}[z^n]G_j(z),
   \qquad G_j(z)=\prod_{i=0}^{q-1}(1-z\omega^{jf(i)})^{-1}.
\]
The \(j=0\) contribution is the stated average. For \(j\ne0\), let
\(d=m/\gcd(j,m)>1\), the order of \(\omega^j\). Two labels give the
same pole precisely when they coincide modulo \(d\). For \(q=2\), where
\(m=b\), the labels \(0,1\) are distinct modulo every such \(d\), so all
poles are simple and \(E=0\).

For \(q\ge3\), the identity \(m=t b^{q-2}+1\) makes both \(b\) and \(t\)
units modulo \(d\). Hence no power of \(b\) is zero, and adjacent powers
cannot coincide: their difference is \(b^it\). Repetitions among
\(1,b,\ldots,b^{q-2}\) occur exactly at index differences divisible by
\(\operatorname{ord}_d(b)\). If a collision occurs, the smallest attainable
value of \(\operatorname{ord}_d(b)\) is \(\ell\): modulo a prime divisor of \(\gcd(m,b^\ell-1)\), \(b\) has order
dividing \(\ell\), and any smaller order would contradict the definition of \(\ell\).
Since one residue class modulo an order \(o\ge\ell\) contains at most \(\lfloor(q-2)/o\rfloor+1\le E+1\) of the indices \(0,\ldots,q-2\), with equality for \(o=\ell\), the largest pole order among the \(G_j\) with \(j\ne0\) is \(E+1\); when no collision
occurs, all poles are simple and \(E=0\). The zero label has its own
simple pole.

For the parity estimates, among the \(q-1\) powers
\(1,b,\ldots,b^{q-2}\), a single pole can therefore come from at most
\(\lceil(q-1)/2\rceil\) indices. For even \(q\ge4\), powers two places apart
cannot coincide either. Indeed, \(b^2\equiv1\pmod d\) would give
\(b^{q-2}\equiv1\pmod d\), hence \(0\equiv m\equiv t+1=b\pmod d\),
contrary to \(b\) being a unit. Thus repetitions are at least three
places apart, and no pole has multiplicity greater than
\(\lceil(q-1)/3\rceil\). This gives \(E=0\) for \(q=3,4\); the case \(q=6\) is completed after this proof.

Every pole of \(G_j\) has order at most \(E+1\) and lies on the unit circle.
In the partial-fraction expansion, a term \(c(1-\lambda z)^{-k}\)
contributes \(c\lambda^n\binom{n+k-1}{k-1}=O_{q,t}(n^E)\).
Summing over the fixed set of \(j\ne0\) proves the first estimate; the
constant is independent of \(a\) since \(|\omega^{-aj}|=1\).
The average is of order \(n^{q-1}\), so the relative error is
\(O_{q,t}(n^{E-q+1})\). Taking logarithms gives the redundancy estimate.

Since every such \(\lambda\) satisfies \(\lambda^m=1\), the coefficients of each \(G_j\) with \(j\ne0\) are
polynomials of degree at most \(E\) on each residue class of \(n\) modulo \(m\).
If \(E>0\), we argue as for attainment in the proof of
Theorem~\ref{thm:orbit-balance}: choose a character attaining pole order \(E+1\).
Its generating function has numerator \(1\), so that pole cannot cancel.
Its coefficient polynomial must therefore have degree \(E\) on some
residue class: otherwise its generating function would have denominator
dividing \((1-z^m)^E\). This coefficient is the Fourier transform of the
deviations. The triangle inequality and a finite pigeonhole
argument give a fixed syndrome attaining absolute deviation at least
\(cn^E\) along an infinite subsequence. If \(E=0\), the deviations, which the generating functions define for every \(n\ge0\),
are periodic from \(n=0\), and at \(n=0\) they are not all zero, giving the same
attainment conclusion.
When \(q=4\) the four poles are distinct for every \(j\ne0\), so
\(G_j=P_j(z)/(1-z^m)\) with \(\deg P_j=m-4\). Its coefficients vanish
at the last three residues modulo \(m\), proving exact equality there.
\end{proof}

For \(q=5\), the exact exponent is \(E=1\) if \(t\equiv1\pmod3\), and
\(E=0\) otherwise. Indeed, order \(1\) is impossible, order \(3\) would force \(0\equiv m\equiv t+1=b\pmod d\), and
\(\gcd(m,b^2-1)=\gcd(3,b+1)\). For \(q=6\), \(m=(b-1)b^4+1\) is
odd; orders \(1,2,4\) would force \(0\equiv m\equiv b\pmod d\), contradicting that \(b\) is a unit, while order
\(3\) would force \(0\equiv(b+1)m\equiv2\pmod d\), impossible for odd \(d>1\). Thus no \(d>1\mid m\)
allows a collision, and \(E=0\) for every \(t\ge1\); this completes the
case \(q=6\) of the boundedness claim in Proposition~\ref{prop:cyclic-balance}.
As for four symbols, the simple-pole factorization shows that all six-symbol
classes are equal when \(n\equiv-1,\ldots,-5\pmod m\).

\begin{remark}[Choice of residue]
\label{rem:choice-of-residue}
Class counts at finitely many lengths determine the largest classes at every length.
For the geometric labels and for the radix labels write \(N\) for the modulus, \(e\) for the deviation
exponent (\(u-1\) in Theorem~\ref{thm:orbit-balance}, \(E\) in Proposition~\ref{prop:cyclic-balance}), and \(C_a(n)\) for the size of class \(a\). Fix \(0\le r<N\).
The character expansions above show that
\(C_a(r+kN)-C_0(r+kN)\) is a polynomial in \(k\) of degree at most \(e\)
for every \(k\ge0\): subtracting cancels the common average. Its values
at \(k=0,\ldots,e\) determine it, so counts for lengths
\(0,\ldots,(e+1)N-1\) determine all maximizing classes at every length.
For \(e=0\) one may simply store a maximizing syndrome for each of the
\(N\) residues. This preprocessing depends on the labels and can be large;
it concerns choosing a code, not decoding a received word.
\end{remark}

\begin{example}[Three symbols]
\label{ex:ternary-cyclic}
Here \(m=t^2+t+1\). Subtracting \(1\) from the labels \((1,t+1,0)\), which lowers every syndrome by \(n\),
and interchanging the last two symbols gives an equivalent congruence without \(x_0\) and with the same class sizes:
\[
   -x_1+t x_2\equiv a'\pmod{t^2+t+1},\qquad a'=a-n\pmod m.
\]
At \(t=2\) and \(n=8\), its class sizes, in order \(a'=0,\ldots,6\), are
\((7,6,7,6,6,6,7)\); at \(n=12\), all seven sizes are \(13\).
In the ternary case all nontrivial poles in the proof above are simple,
and their reciprocals are \(m\)-th roots of unity, so the deviation from
the average is periodic in \(n\), of period dividing \(m\).
More specifically, the product of the three distinct factors in the
denominator of \(G_j\) divides \(1-z^m\). Thus
\(G_j=P_j(z)/(1-z^m)\) with \(\deg P_j=m-3\) for \(j\ne0\), so all classes
are exactly equal when \(n\equiv-1,-2\pmod m\).
\end{example}

For one deletion the radix construction has \(m=2^{q-2}+1\), whereas one
can simply use the \(q\) distinct labels \(0,\ldots,q-1\) in \(\mathbb Z_q\).
These choices agree up to a shift of labels for \(q=2,3\); for \(q=3\),
with radix labels \(1,2,0\) in that order, \(L(\mathbf x)\equiv n+x_1+2x_2\pmod3\).
For \(q\ge4\) the radix modulus is strictly larger than \(q\), so its
averaging guarantee is weaker in this regime.

The modulus \(m\) was chosen so that no deletion sum
wraps around. For four symbols and \(t\ge2\), allowing just the largest sum to wrap gives a smaller
modulus that keeps digit decoding and bounded deviations for every
prescribed class.

\begin{theorem}[A smaller radix modulus for four symbols]\label{thm:four-radix-modulus}
Let \(q=4\), \(2\le t\le n\), and \(b=t+1\). Keep the labels
\(0,1,b,b^2\), and use the modulus
\[
  N=t^3+t^2+c,\qquad
  c=\begin{cases}3,&t\equiv1\pmod3,\\1,&t\not\equiv1\pmod3.\end{cases}
\]
Every syndrome class modulo \(N\) corrects up to \(t\) deletions, and
for fixed \(t\) its size is \(V_{n,4}/N+O_t(1)\), uniformly in the syndrome.
Moreover, \(N<M<m\), where \(M=t^3+t^2+t+1\) is the geometric modulus
and \(m=t(t+1)^2+1\) is the original radix modulus.
\end{theorem}

\begin{proof}
The largest sum of at most \(t\) deleted labels is \(tb^2\);
every other sum is at most \((t-1)b^2+b=t^3+t^2<N\).
The residue of \(tb^2\) is \(tb^2-N=tb-c\). In base \(b\) this is
\((t-1)b+(b-c)\), whose digits sum to \(2t-c>t\) (in both cases \(c<t\), since \(c=3\) forces \(t\ge4\)).
Thus \(tb-c\) cannot itself be a deletion sum: each nonzero deleted label
contributes one digit unit, and at most \(t=b-1\) such labels
produce no carry. Given the syndrome difference, replace \(tb-c\) by \(tb^2\)
if it occurs and otherwise keep its least nonnegative residue.
Its base-\(b\) digits now give the three nonzero-label counts, and the
known number of deletions gives the zero-label count. This is the decoder
of Lemma~\ref{lem:injective-syndrome} with one extra comparison.

For the class sizes, every prime factor of a difference between two labels
divides \(t(t+1)(t+2)\), since \(b-1=t\) and \(b^2-1=t(t+2)\).
Modulo \(t\) and modulo \(t+1\), the new modulus \(N\) is congruent to
\(c\); modulo \(t+2\) it is congruent to \(c-4\).
If \(c=1\), then \(3\nmid t+2\); if \(c=3\), neither \(t\) nor \(t+1\)
is divisible by \(3\), and \(c-4=-1\). (For \(t\equiv1\pmod3\), the value \(c=1\) would make \(3\) divide both \(N\) and \(t+2\), and \(c=2\) always makes both \(N\) and \(t(t+1)\) even.) Thus every label difference is
coprime to \(N\). In the character sum used in
Proposition~\ref{prop:cyclic-balance}, every nontrivial character modulo
\(N\) therefore takes four distinct values on the labels. Its poles are
simple, so its contribution is bounded independently of \(n\).
Finally, \(M-N=t+1-c>0\) and \(m-M=t^2\).
\end{proof}

The smaller modulus helps even when the syndrome is prescribed: for fixed \(t\)
the limiting fraction of all words in each new class is \(1/N\), larger
than both \(1/M\) and \(1/m\) for the original classes
(Theorem~\ref{thm:orbit-balance} and Proposition~\ref{prop:cyclic-balance}).
Thus for all sufficiently large lengths
even the smallest new class carries more messages than the largest class
in either original construction. For \(t=2\) the labels are \(0,1,3,9\)
and \(N=13\), giving the planar difference set used in
\cite[Example~2]{KovacevicTan2018}. Here the only sum that wraps
is \(9+9=18\): residue \(5\) means that two symbols labeled \(9\) were deleted.

\subsection{Comparison of syndrome constructions}
\label{subsec:cyclic-comparison}

For a common measure, let \(N\) be the order of the chosen syndrome group,
or the index of a coding lattice in \(A^{q-1}\). Averaging over the \(N\)
classes, some of which may be empty at a given blocklength, guarantees a class of size at least
\(\lceil V_{n,q}/N\rceil\), hence redundancy at most \(\log_q N\). A smaller
\(N\) improves the unrounded bound \(V_{n,q}/N\), although the two
integer guarantees can coincide after taking ceilings at a particular length.

For \(t\ge2\), the radix and geometric constructions have the same label
set and modulus when \(q=2,3\).
For \(q\ge4\), the binomial expansion gives
\(\sum_{j=0}^{q-2}t^j<(t+1)^{q-2}\), hence
\(M=1+t\sum_{j=0}^{q-2}t^j<1+t(t+1)^{q-2}=m\). More is true
if \(q\ge4\) and \(t\ge2\) are fixed: by Theorem~\ref{thm:orbit-balance} and
Proposition~\ref{prop:cyclic-balance}, the sizes of the geometric and
radix classes, divided by \(V_{n,q}\), tend uniformly to \(1/M\) and
\(1/m\), respectively. Thus for all sufficiently large \(n\) even the
smallest class from the geometric labels is larger than the largest
radix class modulo \(m\). For \(q=5,t=2\) their size ratio tends to \(55/31\);
every class from the geometric labels has size exactly \(V_{n,5}/31\) for
\(n\equiv27,28,29,30\bmod31\).

For \(t\ge2\) and even \(q\ge4\), the radix classes have the smaller deviation
exponent about their own average. For every radix class the deviation is
\(O_{q,t}(n^{\lceil(q-1)/3\rceil-1})\)
(Proposition~\ref{prop:cyclic-balance}). In contrast, for the geometric
labels some fixed class has deviation of order \(n^{q/2-1}\) along an
infinite sequence of lengths. Exact balance gives every prescribed syndrome
the same number of messages, with no search for a largest class, but it
need not give more messages. For example, for \(q=4,t=2,n=16\)
the \(19\) radix classes each have \(51\) words (since \(16=19-3\)).
The \(969\) words cannot split equally among the \(15\) geometric
classes. In fact, the coefficient formula in the proof of
Theorem~\ref{thm:orbit-balance} gives five
classes of size \(63\), six of size \(65\), and four of size \(66\).
Thus even the smallest of these classes carries more messages than any
one of the \(19\) equal classes.

Even at lengths without equality, periodic deviations simplify the choice
of a largest class: adding one period changes every class size by the same
amount, so the maximizing syndromes, including ties, are unchanged.
Thus one period of lengths suffices for the geometric labels at prime
\(q\), and for the radix labels whenever their exponent is zero.
Remark~\ref{rem:choice-of-residue} gives the corresponding finite
calculation for both families when the deviations are not bounded.

\paragraph{A smaller prime modulus for six or more symbols.}
We start from five integer labels, adjoin a sixth, extend by \(L\mapsto tL+1\) up to \(q\) labels, and take a prime modulus just above \(t\) times the largest label, which keeps all deviations bounded. If \(L\) is the largest of the five labels, one can always adjoin
\(tL+1\) as the sixth: division by \(tL+1\) recovers its multiplicity
because the sum of all smaller deleted labels is at most \(tL\).
A smaller sixth label improves every subsequent modulus. We use the
first five labels of the greedy \(B_h\)-set, which starts with \(0,1\) and repeatedly adjoins the least larger integer that keeps the \(B_h\) property; for \(6\mid h\) they are the labels \(0,1,B,C,A\) of the proof below, the last of them determined by
Nathanson and O'Bryant \cite[Theorem~1]{NathansonOBryant2024}. We adjoin as
sixth label the value \(U\) that
O'Bryant~\cite[Conjecture~9]{OBryant2025} conjectures to be the next element of this greedy set when \(6\mid h\). For the coding application it
suffices to prove that this label works when \(6\mid h\): a \(B_h\)-set
is also a \(B_t\)-set for every \(t\le h\).

\begin{proposition}\label{prop:six-label-prime}
For every fixed \(q\ge6\) there are explicit integer labels and a prime modulus
\(P=(1/3+o(1))t^{q-1}\) as \(t\to\infty\) such that every syndrome class modulo \(P\)
corrects up to \(t\) deletions.
For fixed \(q,t\), every syndrome class has size
\(V_{n,q}/P+O_{q,t}(1)\), uniformly in the syndrome and blocklength \(n\ge t\).
\end{proposition}
\begin{proof}
Put \(k=\lceil t/6\rceil\), \(h=6k\), and set
\[
 B=h+1,\qquad C=hB+1,\qquad S=1+B+C,\qquad A=3kS+1,\qquad U=4kA+kS+1.
\]
The five known labels are \(0,1,B,C,A\); we show that adjoining \(U\)
preserves uniqueness of all sums of \(h\) labels. If two different sums
agree, their numbers of copies of \(U\) differ, since the old labels
form a \(B_h\)-set. Their difference is exactly one, since \(2U>hA\).
Cancelling common terms gives
\(U=d_4A+d_3C+d_2B+d_1\), where the positive coefficients total at
most \(h\) and the absolute values of the negative ones total at most
\(h-1\): the side containing the extra \(U\) has one fewer place.

Write \(d_4=4k+i\). The remaining
signed sum has absolute value at most \(hC<2A\), and \(4k<U/A<4k+1\);
thus \(i\in\{-1,0,1,2\}\). The array below lists, for each \(i\), the expansion \(U-d_4A=xC+yB+z\)
with \(0\le yB+z<C\) and \(1\le z<B\); such an expansion exists in all four rows. The coefficient totals give \(-(h-1)B\le d_2B+d_1\le hB<C\), so the only
possibilities for \(d_3\) are \(x+b\) with a carry \(b\in\{0,1\}\) at the place of \(C\); in this proof \(b\) and \(c\) denote carries. Since \(z\ge1\), we have \(z-b\ge0\); with
\(C=hB+1\) and \(|d_1|\le h<B\) this leaves a second carry \(c\in\{0,1\}\) at the place of \(B\):
\[
 (d_3,d_2,d_1)=(x+b,\ y-hb+c,\ z-b-Bc),\qquad c\in\{0,1\}.
\]
For the pairs in the last column the negative coefficients have total
magnitude at least \(h\); for every other pair the positive coefficients,
including \(d_4\), total at least \(h+1\). For example, when \(i=1\)
and \(b=c=1\), the three lower coefficients all equal \(-2k\).
\[
\begin{array}{c|rrr|l}
 i&x&y&z&(b,c)\text{ excluded by negative coefficients}\\ \hline
 -1&4k&4k&4k+2&\varnothing\\
 0&k&k&k+1&(1,1)\\
 1&-2k-1&4k-1&4k+2&(1,1)\\
 2&-5k-1&k-1&k+1&(0,1),\ (1,0),\ (1,1)
\end{array}
\]
This rules out a collision.

For \(q\ge6\), extend these six labels
\(q-6\) times by \(L\mapsto tL+1\), retaining the earlier labels.
Since \(U=t^4/3+O(t^3)\), the final maximum \(L\) satisfies
\(tL=(1/3+o(1))t^{q-1}\). Take the least prime \(P>tL\); the prime
number theorem gives \(P\sim tL\). No deletion sum wraps modulo \(P\).
Since \(P\) is a prime exceeding every label difference, every nontrivial
character takes distinct values on the labels, so the
simple-pole argument of Theorem~\ref{thm:orbit-balance} gives the
stated bounded deviations.
\end{proof}

For \(q>6\), division by the added labels, from the largest down, first
recovers their multiplicities. For every \(q\ge6\), decoding with the six initial labels then
needs no table of deletion patterns: since \(tA<2U\)
and \(tC<2A\), at most two consecutive multiplicities are possible for
each of \(U\) and \(A\). For each choice, ordinary division recovers the
remaining digits since \(tB<C\) and \(t<B\); the known deletion count
determines the number of zeros. Only one candidate has the required
sum and count.

For comparison, starting the same extension and prime selection at the
fifth label \(A\sim t^3/2\) gives leading coefficient \(1/2\), and the
geometric modulus \(M\) has coefficient \(1\). Thus, first taking long
blocklengths for each fixed \(t\), then increasing \(t\), a prescribed
class carries \(3/2\) times as many messages as with five initial labels,
and three times as many as with the geometric labels, in the limit.
For \(q\ge6\) and large \(t\), the scaled lattice below has a smaller index, asymptotically
\(q2^{1-q}t^{q-1}\), but its quotient need not be cyclic and the difference between its largest and
smallest class sizes grows with \(n\), although more slowly than the class sizes themselves.
For the greedy counterpart of \(U\), the fifth positive element of the
greedy \(B_h\)-set, O'Bryant~\cite{OBryant2025}
proved an upper bound with leading coefficient \(0.467214\); here the
explicit sixth label suffices without determining the greedy minimum.

\paragraph{A scaled lattice.}
Another way to obtain a small syndrome group is to scale the single-deletion
lattice. Take
\[
 \cL=\left\{z\in A^{q-1}:\sum_{i=0}^{q-1}i z_i=0\pmod q\right\},
 \qquad u=\left\lceil\frac{t+1}{2}\right\rceil.
\]
Here \(u\) is unrelated to \(u=q/p\) in Theorem~\ref{thm:orbit-balance}. The syndrome map \(z\mapsto\sum_i iz_i\bmod q\) is onto since \(e_1-e_0\) maps to \(1\), where the \(e_i\)
are the standard unit vectors. Its kernel contains no vector \(z\) with
\(\|z\|_1/2=1\): every such vector is \(e_i-e_j\) for \(i\ne j\), which maps to \(i-j\not\equiv0\pmod q\).
Thus the intersection of the simplex with any translate of \(u\cL\)
has minimum distance at least \(2u\ge t+1\). Since \(\cL\) has
rank \(q-1\), the index is \([A^{q-1}:u\cL]=q u^{q-1}\). For fixed
\(q\ge3\), its ratio to either \(m\) or \(M\) tends to \(q/2^{q-1}<1\)
as \(t\to\infty\). For example, at \((q,t)=(4,3)\) the lattice index is
\(32\), compared with \(M=40\) and \(m=49\).
The quotient for this construction need not be cyclic. For \(q=3\),
dropping coordinate zero gives the lattice generated by \((u,u)\) and
\((0,3u)\) in Kova\v{c}evi\'c and Tan's proof of Theorem~11~\cite{KovacevicTan2018}.

The scaled lattice also has a decoder requiring no table. A class fixes
the coordinate residues \(0\le\rho_i<u\) and a checksum \(a\bmod q\):
its multiplicity vectors are \(\mathbf x=\boldsymbol\rho+u\mathbf z\) with
\(\sum_i i z_i=a\bmod q\). Indeed, two length-\(n\) multiplicity vectors have these
same data precisely when their difference lies in \(u\cL\). Given the output
\(\mathbf y\) after \(r\le t\) deletions, let \(d_i\) be the least
nonnegative residue of \(\rho_i-y_i\) modulo \(u\). The missing
multiplicities are at least \(d_i\) and agree with them modulo \(u\).
Since \(r<2u\), either \(\sum_i d_i=r\) and decoding is finished, or
one symbol is missing a further \(u\) times. In the latter case subtract
the checksum of \((\mathbf y+\mathbf d-\boldsymbol\rho)/u\) from \(a\)
modulo \(q\) to recover that symbol's index.

These classes can be counted as well. Write \(N_a^{(q)}(h)\) for the
number of length-\(h\) multisets whose symbol sum is \(a\bmod q\).
The class just described has size \(N_a^{(q)}(h)\), where
\(h=(n-\sum_i\rho_i)/u\), and is empty if \(h\) is not a nonnegative
integer. At length \(n\) the \(qu^{q-1}\) classes are those with \(\sum_i\rho_i\equiv n\pmod u\); all of them are nonempty once \(n\ge q(u-1)+u\). The largest scaled class has size \(N_0^{(q)}(\lfloor n/u\rfloor)\):
the divisor formula in Proposition~\ref{prop:sum-modulo-balance} is
maximal at \(a=0\), and adjoining a zero preserves this checksum.
Take \(\boldsymbol\rho=(n\bmod u,0,\ldots,0)\) to attain the bound.

For fixed \(q\ge3,\ t\ge2\), the spread of the scaled classes, the difference between the largest and
smallest class sizes, is \(\Theta_{q,t}(n^{q-2})\).
To see the lower bound, put \(w=n\bmod u\) and \(k=\lfloor n/u\rfloor\).
Residue vectors of totals \(w\) and \(w+u\) both exist because
\(w+u\le2u-1\le q(u-1)\). For \(k\ge1\), averaging over the \(q\)
checksums for each vector gives a difference
\[
 \frac{V_{k,q}-V_{k-1,q}}q=\frac1q\binom{k+q-2}{q-2}.
\]
For the upper bound, all quotient lengths differ from \(k\) by a bounded
amount, so their averages differ by \(O_{q,t}(n^{q-2})\).
Proposition~\ref{prop:sum-modulo-balance} bounds the deviations from
each such average by \(O_{q,t}(n^{\lfloor q/2\rfloor-1})\).
The geometric and radix constructions have spread \(o(n^{q-2})\): by
Theorem~\ref{thm:orbit-balance} and Proposition~\ref{prop:cyclic-balance-sharp},
their deviation exponents are at most \(q/2-1\) and
\(\lfloor(q-2)/2\rfloor\), respectively.
Thus a smaller lattice index can come with less uniform class sizes.

\paragraph{Finite-field comparisons.}\label{par:varshamov-power-sum}
In the power-sum construction of Varshamov~\cite{Varshamov1973}
(see also \cite{KiahPuleoMilenkovic2016,MilenkovicPan2024}), the
argument works over any finite field \(\mathbb F_s\) with \(s\ge q\) and
\(\operatorname{char}\mathbb F_s>t\). Choose distinct elements
\(\alpha_i\in\mathbb F_s\) and put \(k=\min\{t,q-1\}\). The label
of symbol \(i\) is \((\alpha_i,\ldots,\alpha_i^k)\) in the additive group
\(\mathbb F_s^k\). For \(t<q-1\), the first \(r\) power sums determine, by Newton's identities, the monic
polynomial whose roots are the \(\alpha_i\) of the \(r\) deleted symbols.
For \(t\ge q-1\), adjoin the known zeroth sum
\(r=\sum_i E(i)\) to the first \(q-1\) sums, where \(E(i)\) is the number of deleted copies of symbol \(i\). The resulting Vandermonde
system is invertible; the field solution lifts uniquely because
\(0\le E(i)\le t<\operatorname{char}\mathbb F_s\). The syndrome group
therefore has order \(s^k\). In the prime-field specialization one may take
a prime \(p>\max\{t,q-1\}\); for fixed \(q\) and large \(t\), a prime
\(t<p<2t\) gives the same \((q-1)\log_q t+O_q(1)\) redundancy bound as
the radix labels. At \(t=1\), one additive check suffices over any
finite field of size at least \(q\); for prime \(q\) and labels
\(\alpha_i=i\), it is the usual sum modulo \(q\).

Bose and Chowla~\cite[Theorems~1 and~2]{BoseChowla1960} give two classical
cyclic \(B_t\)-sets when \(s\) is a prime power and \(t\ge2\). The affine
form has \(s\) labels in \(\mathbb Z_{s^t-1}\), while the projective form has
\(s+1\) labels in
\(\mathbb Z_{(s^{t+1}-1)/(s-1)}\); any required subset of labels may be
retained. For \(t=2\) the projective form is Singer's construction.
Our polynomial constructions~\cite[Theorems~1 and~2]{KreindelEssayagZabokritskiy2026}
allow any \(f\in\mathbb F_s[X]\) with no root in \(\mathbb F_s\), of degree
\(t\) in the affine case and \(t+1\) in the projective case. Their syndromes
belong to the unit group of \(\mathbb F_s[X]/(f)\), modulo nonzero field
scalars in the projective case. Taking \(f\) irreducible in each case gives the two Bose--Chowla
rows of Table~\ref{tab:syndrome-comparison}. That work compares redundancies only for \(t<q\), with \(q=s\) or \(s+1\), respectively; Table~\ref{tab:syndrome-comparison} and the comparisons after it also cover \(t\ge q-1\).

For Lee-sphere packings, whose density Xiao and Zhou~\cite{XiaoZhou2024} study,
Kova\v{c}evi\'c~\cite[proof of Theorem~4.1]{KovacevicCross2022}
combines a \(B_t\)-syndrome with the sum of the coordinates modulo \(2t+1\).
To see that on our simplex this check reduces to the power-sum check, take the labels
\((1,x,\ldots,x^t)\in\mathbb Z_{2t+1}\times\mathbb F_s^t\) for \(q\le s-1\)
distinct nonzero \(x\), with \(\operatorname{char}\mathbb F_s>t+1\); the argument below needs only distinct \(x\) and \(\operatorname{char}\mathbb F_s>t\).
In a signed difference of \(\ell_1\)-norm at most \(2t\), the positive
and negative parts have sizes differing by at most \(2t\). The first
coordinate therefore distinguishes unequal sizes. If the sizes agree,
each is at most \(t\), and Newton's
identities distinguish the two multisets unless they coincide. On the
length-\(n\) simplex the first coordinate is constant and the number of
deletions is known. Dropping it reduces the sufficient syndrome-group
order from \((2t+1)s^t\) to \(s^t\).
What remains is the power-sum check above: it uses at most
\(\min\{t,q-1\}\) field elements, and its labels may include \(0\).

\begin{table}[htbp]
\centering
\caption{Syndrome orders or lattice indices \(N\); each gives the unrounded
redundancy guarantee \(\log_q N\). Class balance is not tabulated; the text treats it for the first five rows and for one projective example. The Lee-check row, listed for comparison with \cite{KovacevicCross2022}, omits the coordinate-sum check, which is constant at a fixed length; what remains is the power-sum check, whose row has weaker hypotheses and order at most \(s^t\).}
\label{tab:syndrome-comparison}
\small
\begin{tabular}{@{}lll@{}}
\hline
Construction & \(N\) & Parameters \\ \hline
Radix, base \(t+1\) & \(t(t+1)^{q-2}+1\) & all \(q,t\) \\
Radix, Theorem~\ref{thm:four-radix-modulus} & \(t^3+t^2+c\) & \(q=4,\ t\ge2\); \(c\) as in the theorem \\
Geometric labels (Nathanson) & \(1+t+\cdots+t^{q-1}\) & \(t\ge2\) \\
Six initial labels, Proposition~\ref{prop:six-label-prime} & \((1/3+o(1))t^{q-1}\) &
  fixed \(q\ge6,\ t\to\infty\) \\
Scaled lattice & \(q\lceil(t+1)/2\rceil^{q-1}\) & all \(q,t\) \\
Finite-field power sums & \(s^{\min(t,q-1)}\) &
  prime power \(s\ge q,\ \operatorname{char}\mathbb F_s>t\) \\
Bose--Chowla (affine) & \(s^t-1\) & prime power \(s\ge q,\ t\ge2\) \\
Bose--Chowla (projective) & \(1+s+\cdots+s^t\) &
  prime power \(s\ge q-1,\ t\ge2\) \\
Lee check, fixed length & \(s^t\) & prime power \(s>q,\ \operatorname{char}\mathbb F_s>t+1\) \\
\hline
\end{tabular}
\end{table}

Table~\ref{tab:syndrome-comparison} compares syndrome-group orders and lattice indices.
For example, at \((q,t)=(6,2)\), the projective construction with \(s=5\)
has syndrome order \(31\), compared with \(48\) for the affine construction
with \(s=7\), \(63\) for the geometric labels, and \(163\) for the
base-three radix labels. The order \(31\) is best possible for a six-element
\(B_2\)-set in any finite Abelian group: the \(6\cdot5\) ordered
nonzero differences must be distinct. Since
\(31\) is prime, every nontrivial character takes distinct values on the
six labels. As in Theorem~\ref{thm:orbit-balance}, the deviations of
all \(31\) syndrome classes are periodic and vanish for
\(n\bmod31\in\{26,27,28,29,30\}\). These are the only residues of
equality: \(31\) divides \(V_{n,6}=\binom{n+5}{5}\) only when it
divides one of \(n+1,\dots,n+5\).

The comparison with finite-field power sums changes with the parameters.
If \(q\ge3\) and \(t\ge q-1\), the binomial expansion gives
\(M<(t+1)^{q-1}\le s^{q-1}\) for every admissible field, so the
geometric labels use a smaller syndrome group than the field construction
in the table. If instead \(t\ge2\) is fixed and \(q\) grows, a prime
\(q<p<2q\) gives a group of order \(p^t=O_t(q^t)\), whereas \(M\) grows
exponentially in \(q\).

With fixed \(q,t\), all the rows permit a table indexed by \(r\) and the
syndrome difference. Without such a table the geometric and radix constructions
use integer digits, whereas the power-sum constructions use field arithmetic.
The affine Bose--Chowla construction also admits a table-free decoder. Take a
primitive element
\(\theta\in\mathbb F_{s^t}\) with monic minimal polynomial
\(\mu\in\mathbb F_s[X]\). For each selected \(\alpha\in\mathbb F_s\),
the label \(\log_\theta(\theta+\alpha)\) is fixed at setup.
If \(\Delta\) is the syndrome difference and \(e_\alpha\) are the missing
multiplicities, then \(\theta^\Delta=F(\theta)\), where
\(F(X)=\prod_\alpha(X+\alpha)^{e_\alpha}\) is monic of known degree \(r\).
Represent \(\theta^\Delta\) by \(R(\theta)\) with \(\deg R<t\).
Since \(F\equiv R\bmod\mu\), we have \(F=R\) for \(r<t\) and \(F=R+\mu\)
for \(r=t\). Repeated division by the known factors \(X+\alpha\) recovers
their multiplicities without a discrete logarithm at decoding.

\paragraph{Comparison with optimal lattice indices.}
Let \(\phi(t,d)\) be the least order of a finite Abelian group containing
a \(B_t\)-set of \(d+1=q\) elements; equivalently, the least index of a lattice in \(A^{q-1}\) without nonzero vectors of \(\tfrac12\ell_1\)-norm at most \(t\). Kova\v{c}evi\'c and Tan
\cite[Theorem~2.1]{KovacevicTan2017} proved that for fixed \(d\),
\[
 \lim_{t\to\infty}\frac{\phi(t,d)}{t^d}
   =\frac{1}{d!\,\delta_L(\Delta^d)},
\]
where \(\delta_L(\Delta^d)\) is the lattice packing density of a
\(d\)-dimensional simplex. Thus \(\log_q\phi(t,q-1)=(q-1)\log_q t+O_q(1)\) for fixed \(q\).
Here \(\phi\) is a minimum over all finite Abelian groups; Kova\v{c}evi\'c and
Tan~\cite[Section~3.3]{KovacevicTan2017} ask whether cyclic groups attain the same leading constant.

For three symbols the limiting density \(\lim_{n\to\infty} S_3(n,t)/V_{n,3}\), over all codes, is known for
every \(t\)~\cite[Corollary~16]{KovacevicTan2018}. At \(t=3,4\) its
reciprocal is respectively \(12,19\), whereas Construction~\ref{const:cyclic}
uses \(13,21\): its limiting code sizes are respectively \(12/13\) and
\(19/21\) of optimum. For four symbols and two deletions the labels
\(\{0,1,3,9\}\) form a planar difference set modulo \(13\): the twelve
ordered differences exhaust the nonzero residues. Since any four-element
\(B_2\)-set requires twelve distinct nonzero differences, no Abelian
group of order below \(13\) contains one. Thus at \(t=2\) the modulus
\(13\) of Theorem~\ref{thm:four-radix-modulus} equals \(\phi(2,3)\),
compared with \(M=15\) and \(m=19\). By \cite[Corollary~16]{KovacevicTan2018},
\(1/13\) is also the limiting density
\(\lim_{n\to\infty}S_4(n,2)/V_{n,4}\) over all codes.

\section{Upper bounds and codes with few surviving symbols}
\label{sec:bounds}

\subsection{Packing bounds and the averaging guarantees}
\label{subsec:general-bounds-comparison}

Fix \(q\ge2\) and \(1\le t<n\). As in Section~\ref{sec:pre}, \(V_{s,q}\) counts
the multisets of size \(s\). Set \(V_{s,q}=0\) for \(s<0\), and put
\(\rho=\lfloor t/2\rfloor\), \(a=\lceil t/2\rceil\).
We compare the following three elementary upper bounds on \(S_q(n,t)\) with
each other and with the averaging guarantees of Section~\ref{sec:cyclic}:
\begin{equation}
\label{eq:three-general-bounds}
   K=\frac{V_{n+t,q}}{V_{t,q}},\qquad
   P_{\rm sp}=\frac{V_{n,q}}{V_{\rho,q}},\qquad
   P_{\rm proj}=V_{n-t,q}.
\end{equation}
Here \(K=K(n,q,t)\) is the insertion bound of Kova\v{c}evi\'c and Tan
\cite[Theorem~17, Eq.~(31)]{KovacevicTan2018}: each codeword has exactly \(V_{t,q}\) supermultisets of size \(n+t\), and two codewords at distance at least \(t+1\) share none. For the sphere-packing
bound \(P_{\rm sp}\), use the closed ball
\(B_\rho(S)=\{T\in\cS_{n,q}:d(S,T)\le\rho\}\).
Distinct codewords have disjoint balls since their distance is
at least \(t+1>2\rho\).
Each ball contains at least \(V_{\rho,q}\) words: retain a fixed
\((n-\rho)\)-submultiset of its center and add an arbitrary \(\rho\)-multiset.
Equality holds at a constant word; see also
\cite[Lemma~5]{KreindelEssayagZabokritskiy2026}. For the projection bound
\(P_{\rm proj}\), choose one \((n-t)\)-submultiset of each codeword.
No two choices can agree.

One further bound changes the comparison at fixed alphabet size.
The code--anticode bound of Kova\v{c}evi\'c and Tan for the simplex
\cite[Theorem~14, Eq.~(21)]{KovacevicTan2018} is
\begin{equation}
\label{eq:KT-anticode}
   S_q(n,t)\le K_{\rm ac}
   :=\frac{V_{n,q}}{\beta(t,q)}+V_{n,q}-V_{n-qa,q}.
\end{equation}
Here \(\beta(t,q)\) is the size of their anticode, the set of vectors in \(\mathbb Z^{q-1}\) whose positive
coordinates sum to at most \(a\) and whose negative coordinates have total
absolute value at most \(\rho\):
\[
   \beta(t,q)=
   \sum_{j=0}^{\min(a,q-1)}
      \binom{q-1}{j}\binom{a}{j}
      \binom{\rho+q-1-j}{q-1-j}.
\]
The summand counts vectors with exactly \(j\) positive coordinates
\cite[Eq.~(19)]{KovacevicTan2018}. To see why this set gives a packing
bound, drop the last coordinate of the multiplicity vector, which is determined by the first
\(q-1\) and the blocklength. A difference vector then has distance equal
to the larger of its positive sum and its absolute negative sum, since the dropped coordinate joins the smaller side.
Any two anticode vectors differ by at most \(a+\rho=t\) in this distance,
so anticode translates about distinct codewords cannot intersect.
At a center whose multiplicities are
all at least \(a\), the full anticode fits inside the simplex; disjoint
translates bound the number of such centers by \(V_{n,q}/\beta(t,q)\).
The difference \(V_{n,q}-V_{n-qa,q}\) in \eqref{eq:KT-anticode} counts all remaining centers.
This correction is \(O_{q,t}(n^{q-2})\) at fixed
\(q,t\), and equals \(V_{n,q}\) when \(n<qa\).

For finite parameters the ratios
\[
 \frac{K}{P_{\rm sp}}
 =\prod_{i=1}^{q-1}\frac{(\rho+i)(n+t+i)}{(t+i)(n+i)},\qquad
 \frac{K}{P_{\rm proj}}
 =\prod_{i=1}^{q-1}\frac{i(n+t+i)}{(t+i)(n-t+i)}
\]
give the exact comparison before rounding. At \((n,q,t)=(14,5,12)\),
\(P_{\rm sp}=3060/210\) gives \(S_5(14,12)\le14\), while \(K\) and
\(P_{\rm proj}\) each give \(15\) after rounding down. At \((4,20,2)\),
\(P_{\rm proj}=210\), compared with \(442.75\) for \(P_{\rm sp}\) and more
than \(843\) for \(K\). In both examples \(n<qa\), so \(K_{\rm ac}>V_{n,q}\). These examples only rank the bounds: both have two survivors, where Theorem~\ref{thm:bounded-multiplicity} below gives much smaller values, for instance \(S_5(14,12)=5\).

To compare these bounds with the code from the geometric labels, assume
\(t\ge2\) and let \(M=1+t+\cdots+t^{q-1}\). Its unrounded averaging guarantee is \(V_{n,q}/M\);
for any of the displayed upper bounds \(U\), the ratio to this guarantee
is \(MU/V_{n,q}\). This ratio bounds the factor by which an optimal
code can exceed the averaging guarantee.
For the radix code of Subsection~\ref{subsec:radix}, including \(t=1\), replace
\(M\) by \(m=t(t+1)^{q-2}+1\) throughout these ratios.
The order of the limits matters:

\begin{itemize}
\item For fixed \(q,t\) and \(n\to\infty\), the four bounds
  \(K,P_{\rm sp},P_{\rm proj},K_{\rm ac}\), divided by \(V_{n,q}\), tend to
  \(1/V_{t,q},1/V_{\rho,q},1,1/\beta(t,q)\), respectively. Thus the limiting
  ratio of their minimum to this guarantee is \(M/\max\{V_{t,q},\beta(t,q)\}\).
  At \(q=3,t=2\), for example, \(\beta(2,3)=7\) whereas \(V_{2,3}=6\).
  Only \eqref{eq:KT-anticode} among these bounds has the optimal leading
  constant \(1/7\) in this example
  \cite[Corollary~16]{KovacevicTan2018}; since \(M=7\) here, the geometric classes are asymptotically optimal. The other limiting densities quoted from that corollary, in Subsection~\ref{subsec:cyclic-comparison} for \(q=3\) with \(t=3,4\) and for \(q=4\) with \(t=2\), and in Section~\ref{sec:elementary-benchmarks} for \(t=1\), also equal \(1/\beta(t,q)\). In these cases \(\beta(t,q)\ge V_{t,q}\), so the limiting ratio above equals the limiting ratio of an optimal code to the guarantee; in general it only bounds that ratio.
\item If one subsequently lets \(t\to\infty\) with \(d=q-1\) fixed, the
  ratio \(M/V_{t,q}\) tends to \(d!\). The anticode comparison instead gives
  \[
     \frac{M}{\beta(t,q)}
       \longrightarrow \frac{2^d d!}{\binom{2d}{d}} .
  \]
  Indeed, in the sum for \(\beta\) the coefficient of \(t^d\) is
  \(\sum_{j=0}^d\binom dj^2/(2^d d!)=\binom{2d}d/(2^d d!)\).
  The same two limits hold with \(m\) in place of \(M\), since both have
leading term \(t^d\).
\item For fixed \(n,t\) and \(q\to\infty\), \(K\), \(P_{\rm sp}\) and
  \(P_{\rm proj}\) have orders \(q^n,q^{n-\rho},q^{n-t}\), respectively,
  and \(K_{\rm ac}\ge V_{n,q}\) eventually. Projection therefore gives the
  smallest of the four, while the modulus \(M\) grows exponentially in \(q\).
  The radix modulus \(m\) also grows exponentially, including at \(t=1\).
  In contrast, the field power sums of
  Subsection~\ref{subsec:cyclic-comparison} use \(O_t(q^t)\) syndromes
  here, giving a class of order \(q^{n-t}\), the same order as projection.
\item If instead \(q\) and the number \(k=n-t\) of survivors are fixed
  and \(n\to\infty\), \(K\) and \(P_{\rm sp}\) both tend to \(2^{q-1}\)
  and \(P_{\rm proj}=V_{k,q}\). For \(k\ge2\), Theorem~\ref{thm:bounded-multiplicity} below gives the exact value \(S_q(n,n-k)=q\) once \(n>q(k-1)\):
  no multiplicity vector of that length has every multiplicity below \(k\).
\end{itemize}

\subsection{Reduction at a fixed number of survivors}
\label{subsec:bounded-multiplicity}

When \(k\) symbols survive, two codewords are confusable precisely when
they have at least \(k\) symbols in common, counted with multiplicity.
For one survivor, distinct codewords must have disjoint supports. The \(q\) constant words
therefore attain \(S_q(n,n-1)=q\) for every \(n\ge1\), the deletion form of
\cite[Theorem~II.1(c)]{ChenMaZhang2021}. For \(k\ge2\), only the \emph{capped} words, those
with all multiplicities below \(k\), require an additional packing.

Let \(A_k(N,d,w)\) be the maximum number of vectors in
\(\{0,\ldots,k-1\}^{N}\) with coordinate sum \(w\) and pairwise
\(\ell_1\)-distance at least \(d\); it is zero if no such vector exists.
This is the notation of Chen, Ma and Zhang for constant-weight codes in the
\(\ell_1\)-metric~\cite[Section~II-A]{ChenMaZhang2021}; their Section~II-C relates such codes to packings.

\begin{theorem}[Bounded-multiplicity reduction]
\label{thm:bounded-multiplicity}
For \(q\ge2\) and \(2\le k\le n\),
\[
   S_q(n,n-k)=q+A_k\bigl(q,2(n-k+1),n\bigr).
\]
\end{theorem}

\begin{proof}
Assign to each codeword with multiplicity at least \(k\) in some coordinate
one such coordinate. Two codewords assigned the same coordinate would share
\(k\) copies of one symbol; hence there are at most \(q\) assigned codewords.
The other codewords are capped, and their pairwise
\(\ell_1\)-distance is at least \(2(n-k+1)\), since it is twice the
deletion distance. Conversely, adjoin all \(q\) constant words to an
optimal code of capped words. Any capped word meets any constant word in
at most \(k-1\) copies, and distinct constant words have no common symbol.
\end{proof}

\paragraph{Two survivors.}
Let \(D(q,n,s)\) be the maximum number of \(n\)-subsets of a \(q\)-set in
which no \(s\)-subset occurs twice, with \(D(q,n,s)=0\) if \(n>q\).
For \(k=2\) the capped words are squarefree, so the theorem reads
\begin{equation}
\label{eq:two-survivors}
   S_q(n,n-2)=q+D(q,n,2),\qquad n\ge2 .
\end{equation}
For \(n=3\) and \(n=4\) these are the weight-three and weight-four formulas
of Chen, Ma and Zhang~\cite[Theorems~III.1 and~III.3]{ChenMaZhang2021},
respectively. Their proofs already count repeated coordinates for the upper
bound and adjoin all constant words to a packing for equality; Theorem~\ref{thm:bounded-multiplicity} extends this argument to every \(k\ge2\) and \(n\ge k\).
In particular the value is \(q\) for \(n>q\) and \(S_4(3,1)=5\): two
triples on four points must share a pair, and one triple can be added to
the four constant words.

\paragraph{Three survivors.}
Setting \(k=3\) gives the ternary reduction
\begin{equation}
\label{eq:three-survivors}
   S_q(n,n-3)=q+A_3(q,2n-4,n),\qquad n\ge3 .
\end{equation}
In Wei and Zhang's notation~\cite[p.~329, Theorem~3.1]{WeiZhang2023},
the ternary code has length \(q\) and constant \(\ell_1\)-weight \(w=n\); its coordinates are
the multiplicities, and the required \(\ell_1\)-distance is at least \(2w-4\).
Their bound, translated into this setting, is \eqref{eq:wei-zhang-translated} below; we derive it directly from the triples that survive. Distinct capped codewords share no triple, so if weights on the triples give every capped word total weight at least \(\binom n3\), the total weight divided by \(\binom n3\) bounds the number of capped codewords.
A capped word with \(i\) doubled
symbols has \(n-i\) distinct symbols. It contains \(\binom{n-i}3\) squarefree
triples and \(i(n-i-1)\) triples with a repeated symbol: choose the doubled
symbol and one other symbol. Give a squarefree triple weight one. To bring
the total weight for this word up to \(\binom n3\), each repeated-symbol
triple needs weight at least
\[
 \frac{\binom n3-\binom{n-i}3}{i(n-i-1)}
 =\frac16\left(3(n-1)+\frac{i^2-1}{n-i-1}\right)
 \qquad (1\le i\le j=\lfloor n/2\rfloor).
\]
The ratio increases with \(i\), since \(i^2-1\) increases and the positive
denominator \(n-i-1\) decreases. Thus the single weight \(\varepsilon/6\),
where \(\varepsilon=3(n-1)+(j^2-1)/(n-j-1)\), suffices for every word;
when \(i=0\) the squarefree triples already have total weight \(\binom n3\).
The available weight is
\(\binom q3+(\varepsilon/6)q(q-1)\), so division by \(\binom n3\) and
\eqref{eq:three-survivors} give
\begin{equation}
\label{eq:wei-zhang-translated}
   S_q(n,n-3)\le q+
   \left\lfloor\frac{q(q-1)(q-2+\varepsilon)}
                          {n(n-1)(n-2)}\right\rfloor .
\end{equation}
At \(n=3\), where \(t=0\), both triple weights are one, and the bound equals the projection
bound \(P_{\rm proj}=V_{3,q}=q+q(q-1)+\binom q3\). For \(n\ge4\) it strictly improves
on the projection bound: after division by \(\binom n3\) a squarefree
triple costs at most \(1/4\), and a repeated-symbol triple costs at most
\(1/[j(n-j-1)]\le1/2\), by the displayed ratio at \(i=j\).
At \(n=4\) the exact value is known:
\(S_q(4,1)=q+\binom q2+D(q,4,3)\)
by Chen, Ma and Zhang~\cite[Theorem~III.2]{ChenMaZhang2021}; Section~\ref{sec:elementary-benchmarks} compares it with \eqref{eq:wei-zhang-translated} (Proposition~\ref{prop:length-four-exact}).

For fixed \(n\ge4\), Wei and Zhang also give the leading term of the
ternary optimum immediately after their Theorem~3.1. Consequently
\begin{equation}
\label{eq:three-survivors-asymptotic}
   S_q(n,n-3)=(1+o(1))\frac{\binom q3}{\binom n3}
   \qquad(q\to\infty).
\end{equation}
This also follows without their result: the squarefree words give the lower bound \(q+D(q,n,3)\), where R\"odl's packing
theorem~\cite{Rodl1985} gives \(D(q,n,3)=(1-o(1))\binom q3/\binom n3\),
and \eqref{eq:wei-zhang-translated} gives the matching upper bound.

\begin{corollary}[Three survivors at and beyond \(n=q+1\)]
\label{cor:t-n-3-exact-large}
For \(q\ge2\) and \(n\ge q+1\),
\[
 S_q(n,n-3)=
 \begin{cases}
  \displaystyle q+\left\lfloor q/\lceil(q-1)/2\rceil\right\rfloor,
      & n=q+1,\\
  q+1,&q+2\le n\le2q,\\
  q,&n\ge2q+1.
 \end{cases}
\]
At \(n=q+1\) the value is \(q+2\), except that \(S_3(4,1)=6\).
\end{corollary}

\begin{proof}
For \(n\le2q\), complement a capped vector coordinatewise, \(x_i\mapsto2-x_i\).
This preserves \(\ell_1\)-distance and leaves \(\ell_1\)-weight \(w=2q-n\). For \(n=q+1\)
the required \(\ell_1\)-distance \(2n-4\) is \(2w\): supports must be disjoint, so
the ternary optimum is \(\lfloor q/\lceil w/2\rceil\rfloor\)
\cite[Theorem~II.1(c)]{ChenMaZhang2021}. For \(q+2\le n\le2q\),
the required distance exceeds \(2w\), so exactly one capped word is possible
\cite[Theorem~II.1(b)]{ChenMaZhang2021}. Above \(2q\) no capped word exists.
Adding the constant words in \eqref{eq:three-survivors} gives the result.
\end{proof}

\section{Single deletions and finite-length optimality}\label{sec:elementary-benchmarks}

For one deletion the geometric labels of Section~\ref{sec:cyclic}, written
as \(b_i=1+t+\cdots+t^{i-1}\) with \(b_0=0\), become \(b_i=i\)
and \(M=q\) at \(t=1\). For a residue \(a\), put
\[
 \cC_{\rm sum}(a)=\{\mathbf x\in\cS_{n,q}:\textstyle\sum_{i=0}^{q-1}i x_i=a\bmod q\},\qquad
 N_a^{(q)}(n)=|\cC_{\rm sum}(a)|.
\]
We call the sets \(\cC_{\rm sum}(a)\) the cyclic syndrome classes, or simply the cyclic classes. The difference between the prescribed and received syndromes gives the
deleted symbol. The question is how close the largest class comes to the
optimum among \emph{all} codes. We first count the classes, then derive
upper bounds valid for every one-deletion code.

\subsection{Counting multisets by the sum of their symbols}
Counting \(N_a^{(q)}(n)\) is exactly the multisubset sum problem in
\(\mathbb Z_q\). The divisor formula below is recorded by
Panyushev~\cite[Eq.~(2.2)]{Panyushev2011}. It is also a specialization
of the count for any finite Abelian group due to Muratovi\'c-Ribi\'c and
Wang~\cite[Corollary~2]{MuratovicWang2015}: a multiset of \(n\) elements of
\(\mathbb Z_q\) is a multiset of at most \(n\) nonzero elements padded with zeros. We recall the cyclic calculation to compare
these classes with optimal codes.

\begin{proposition}[Cyclic specialization of the multisubset count]\label{prop:sum-modulo-balance}
Put \(c_\ell(a)=\sum_{1\le j\le\ell,\ \gcd(j,\ell)=1}\exp(-2\pi i j a/\ell)\).
Then
\[
 N_a^{(q)}(n)=\frac1q\sum_{\ell\mid\gcd(q,n)}
 c_\ell(a)\binom{(n+q)/\ell-1}{q/\ell-1}.
\]
In particular \(N_a^{(q)}(n)=V_{n,q}/q+O_q(n^{d-1})\), where \(d\)
is the largest proper divisor of \(q\). For prime \(q\) the largest class has
size \(\lceil V_{n,q}/q\rceil\).
\end{proposition}

\begin{proof}
Apply the character filter used in the proof of Theorem~\ref{thm:orbit-balance} to all the
elements of \(\mathbb Z_q\). A character of order \(\ell\) takes each
\(\ell\)-th root of unity as its value \(q/\ell\) times. Its generating
function is therefore \((1-z^\ell)^{-q/\ell}\), whose coefficient at \(z^n\)
is zero unless \(\ell\mid n\), and is the displayed binomial coefficient
otherwise. The weights of the characters of order \(\ell\) sum to \(c_\ell(a)\).
The term \(\ell=1\) gives the average; for the others \(q/\ell\le d\).
When \(q\) is prime and \(q\nmid n\), all classes have exactly the average
size. If \(q\mid n\), the zero class has size \((V_{n,q}+q-1)/q\);
all other classes have size \((V_{n,q}-1)/q\).
\end{proof}

Since \(c_\ell(a)\le c_\ell(0)\) and the binomial coefficients are positive, the zero class is always a largest class. This known count also gives the sharp deviation exponent at \(t=1\).
Indeed, for \(a=0\) and \(q\mid n\) all summands are positive. Among the
divisors \(\ell>1\), the smallest prime divisor \(p\) of \(q\) gives the
largest degree, \(q/p-1\). Theorem~\ref{thm:orbit-balance} obtains the
same exponent for the \(q\) geometric labels in the larger group at \(t\ge2\).

For fixed \(q\), Kova\v{c}evi\'c and Tan~\cite[Corollary~16]{KovacevicTan2018}
already give \(S_q(n,1)\sim V_{n,q}/q\). Their Theorem~17, Eq.~(31),
which is the insertion bound \(K\) of \eqref{eq:three-general-bounds} at \(t=1\),
gives the finite-length upper bound \(\binom{n+q}{q-1}/q=V_{n+1,q}/q\) for \(n>1\).
Together with the count above this shows that every prescribed cyclic class
attains the optimal leading term. In the opposite
regime, with \(n\ge7\) fixed and \(q\to\infty\), the unrestricted optimum
satisfies \(S_q(n,1)=V_{n,q}/q+O_n(q^{n-4})\)
by~\cite[Theorem~9.8]{ZabokritskiyWeightedLocalCovers2026}; for prime \(q>n\) every cyclic class has size exactly \(V_{n,q}/q\) by Proposition~\ref{prop:sum-modulo-balance}, so these classes are optimal up to that error term.

\subsection{Counting outputs by support}
Let \(\partial_1(S)\) be the set of distinct outputs obtained by deleting one
symbol of \(S\). Its size is the support size, since deleting two different
symbols gives different outputs. These sets are disjoint for codewords of a
one-deletion code. There is a total budget of \(B=V_{n-1,q}\) outputs.
Using the number of outputs of every input, namely its support size, gives the multiset specialization of
the degree-sequence bound of Cullina and Kiyavash~\cite[preprint,
Sec.~IV-A, Def.~6]{CullinaKiyavash2016}.

\begin{proposition}[Degree-sequence bound]\label{prop:support-shadow-single}
For \(1\le r\le q\), put \(A_r=\binom qr\binom{n-1}{r-1}\), with out-of-range
binomial coefficients understood as zero. Let \(s\) be the first index for
which \(\sum_{r\le s}r A_r>B\). Then
\[
 S_q(n,1)\le \sum_{r<s}A_r+
 \left\lfloor\frac{B-\sum_{r<s}r A_r}{s}\right\rfloor .
\]
\end{proposition}

\begin{proof}
There are \(A_r\) words of support \(r\): choose the symbols and then their
positive multiplicities. If \(c_r\) of these words belong to the code, then
\(0\le c_r\le A_r\) and \(\sum_r r c_r\le B\). To maximize \(\sum_r c_r\)
under just these conditions, take the cheapest words first: exchanging a
chosen word for a cheaper available one never increases the cost. This gives
the stated bound. The cutoff exists because the total cost of all words is
\(qB\): each output extends to exactly \(q\) words.
\end{proof}

For every \(1\le\ell\le q\), these same constraints imply
\(\sum_r c_r\le\sum_{r<\ell}A_r+B/\ell\). Thus the proposition refines
the deletion-side threshold bound of Kova\v{c}evi\'c and
Tan~\cite[Theorem~17, Eq.~(30), with their parameters \(h=r=1\)]{KovacevicTan2018}.
At \(n=4\) it gives \(6\) for \(q=3\) and \(20\) for \(q=5\), whereas the
deletion-side threshold bound, optimized over \(\ell\), gives \(8\) and \(22\), respectively, and
the insertion bound gives \(7\) and \(25\) after rounding down.
The value \(6\) is attained by the three constant ternary words and the
three permutations of \((2,2,0)\). The degree-sequence bound need not beat the insertion bound: at \((q,n)=(3,8)\) they give \(19\) and \(18\).
For fixed \(q\ge3\) the cutoff is eventually \(s=q\): smaller supports cost
\(O_q(n^{q-2})=o(B)\). Since \(\sum_r A_r=V_{n,q}\) and \(\sum_r rA_r=qB\),
once \(s=q\) the unrounded degree-sequence bound is \(U=V_{n,q}-(q-1)B/q\). Thus
\[
 U-\frac{V_{n+1,q}}q
   =\frac{(q-2)n}{q(n+1)}\binom{n+q-2}{q-2}\longrightarrow\infty,
\]
so insertion is eventually stronger even after rounding.

A different use of support is to weight the \emph{outputs}, as in the
fractional covering method~\cite{FazeliVardyYaakobi2015,CullinaKiyavash2016}.
A \emph{fractional cover} assigns nonnegative weights to the outputs so that
the outputs of every word have total weight at least one; since distinct
codewords have disjoint output sets, its total weight bounds \(S_q(n,1)\) from above.
For \(n\ge2\), give each output with \(r\) distinct symbols weight \(1/r\).
A word on \(j\) symbols has \(j\) outputs, each of support at most \(j\),
so their weights sum to at least one. Disjointness now gives
\[
 S_q(n,1)\le\left\lfloor\sum_{r=1}^{\min(q,n-1)}
             \frac1r\binom qr\binom{n-2}{r-1}\right\rfloor .
\]
Indeed, \(B_r=\binom qr\binom{n-2}{r-1}\) counts the outputs of support
\(r\). This bound is no larger than Proposition~\ref{prop:support-shadow-single}:
the counts \(c_r=B_r/r\) are feasible in its continuous relaxation since
\(B_r\le A_r\) and \(\sum_r B_r=B\). Their total cannot exceed the relaxed
maximum, whose floor gives the proposition because the greedy solution
has at most one nonintegral count.

For fixed \(n\ge3\) and growing \(q\) the comparison of the degree-sequence bound with insertion reverses. Supports
through \(n-2\) cost \(o(B)\), but \((n-1)A_{n-1}\sim(n-1)^2B\), so
\(s=n-1\) eventually. Both the unrounded degree-sequence bound and the bound from the weights \(1/r\) (by its
\(r=n-1\) term) are asymptotic to \(B/(n-1)\): they improve on projection
by a factor tending to \(n-1\) and have order \(q^{n-1}\), versus \(q^n\) for insertion.
Relative to \(V_{n,q}/q\), the limiting factor for the degree-sequence
bound is \(n/(n-1)\), as already noted in
\cite[Section~9.2, following Eq.~(66)]{ZabokritskiyWeightedLocalCovers2026}.

At \(n=4,\ q\ge4\) we determine the minimum total weight of a fractional cover.
For outputs of multiplicity types \(3,\ 2+1,\ 1+1+1\) (the partitions of \(3\) formed by the positive multiplicities), use weights
\(1,\tfrac12,\tfrac14\), respectively. The input types
\(4,\ 3+1,\ 2+2,\ 2+1+1,\ 1+1+1+1\) then receive total weights
\(1,\tfrac32,1,\tfrac54,1\), so the cover is feasible. Conversely,
averaging any fractional cover over alphabet permutations preserves its
total weight and makes the weights constant on these three types. Inputs of
types \(4,\ 2+2,\ 1+1+1+1\) force the three weights to be at least
\(1,\tfrac12,\tfrac14\), respectively. Thus the minimum is
\(q+\tfrac12q(q-1)+\tfrac14\binom q3\), which is also the unrounded value of
\eqref{eq:wei-zhang-translated} at \(n=4\). For \(q=5\) the degree-sequence
bound is \(20\); the weights \(1/r\) give \(\lfloor5+20/2+10/3\rfloor=18\),
and the minimum fractional cover gives \(\lfloor5+20/2+10/4\rfloor=17\).
The exact size is \(16\)~\cite[Theorem~III.2]{ChenMaZhang2021}; no choice
of fractional output weights alone can establish it.

\begin{proposition}[Exactness at length four]\label{prop:length-four-exact}
Let \(q\ge4\). Rounded down, the minimum total weight of a fractional cover at \(n=4\) equals \(S_q(4,1)\) precisely when a Steiner quadruple system of order \(q\) exists, that is, for \(q\equiv2,4\pmod6\) by Hanani's theorem~\cite{Hanani1960}.
\end{proposition}
\begin{proof}
Let \(L\) count the
uncovered triples in a largest packing of quadruples with no triple
repeated. Each quadruple covers four triples, so \(D(q,4,3)=(\binom q3-L)/4\),
and the exact formula \(S_q(4,1)=q+\binom q2+D(q,4,3)\) of
Chen, Ma and Zhang~\cite[Theorem~III.2]{ChenMaZhang2021} gives an
unrounded gap of \(L/4\). Rounding gives equality exactly when \(L\le3\), but a
nonempty set of uncovered triples always has at least four members.
Indeed, each pair occurs in \(q-2\) triples, and each quadruple containing
it covers two. For odd \(q\ge5\) every pair must therefore occur in an
uncovered triple; three triples cannot account for all \(\binom q2\ge10\)
pairs. For even \(q\), each pair occurs an even number of times among the
uncovered triples, so an uncovered triple forces three distinct others, one
through each pair. Thus the bound is exact precisely when \(L=0\), that is, when a Steiner
quadruple system of order \(q\) exists.
\end{proof}

In the graph that joins two multiplicity vectors at \(\ell_1\)-distance two, the words
containing a fixed output form a clique, so these output weights form a fractional cover of this graph by
such cliques alone.

At \(n=3\) the degree-sequence bound and fractional covers already give
different bounds for \(q\ge4\). For \(q\ge3\) Proposition~\ref{prop:support-shadow-single} gives
\(q+\lfloor q(q-1)/4\rfloor\), compared with
\(q+\lfloor q(q+1)/4\rfloor\) from threshold \(\ell=2\). Giving a repeated
pair weight one and a distinct pair weight \(1/3\) instead gives
\(S_q(3,1)\le q+\lfloor q(q-1)/6\rfloor\): a nonsquarefree triple
contains a repeated pair, and a squarefree triple contains three distinct
pairs. We call this the pair-counting bound; by \eqref{eq:two-survivors}
the exact value is \(q+D(q,3,2)\). In explicit form, Chen, Ma and
Zhang~\cite[Theorem~III.1 and Lemma~II.1]{ChenMaZhang2021} give
\[
 S_q(3,1)=q+
 \left\lfloor\frac q3\left\lfloor\frac{q-1}{2}\right\rfloor\right\rfloor
 -\mathbf 1_{\{q\equiv5\bmod 6\}}.
\]
For example, at \(q=5\) two triples can be packed but not three, so the
optimum is \(7\), one below the pair-counting bound. At \(q=7\) the Fano
plane and the seven constant words attain \(14\), whereas the
degree-sequence bound is \(17\).

\subsection{Exact comparisons for small alphabets}
Over two symbols the multiplicity of the second symbol identifies a word
with an integer in \([0,n]\), and deletion distance is ordinary distance.
Spacing these integers \(t+1\) apart gives \(S_2(n,t)=1+\lfloor n/(t+1)\rfloor\);
no larger family fits in the interval. Equivalently one takes a largest
residue class modulo \(t+1\), the binary instance of Construction~\ref{const:cyclic};
compare~\cite[Proposition~3.1]{KovacevicVukobratovic2015}.

For one deletion over any alphabet, join two multiplicity vectors when their
\(\ell_1\)-distance is two, that is, when one copy of one symbol is replaced
by a copy of another. Independent sets in the resulting graph are
exactly the one-deletion codes. For \(q=3,4\), the graphs are the triangular
and tetrahedral grids studied by Geramita, Gregory and Roberts. The following
values translate their independence numbers, and the uniqueness statements translate
Machacek's theorems; uniqueness is over a fixed labeled alphabet. The code comparisons identify which syndrome
class realizes each value.

\paragraph{Three symbols.}
By~\cite[Theorem~5.4(2)]{GeramitaGregoryRoberts1986},
\(S_3(n,1)=\lceil V_{n,3}/3\rceil\) except that \(S_3(2,1)=3\) and
\(S_3(4,1)=6\). When \(3\nmid n\), all three cyclic classes have size
\(V_{n,3}/3\); they are optimal except at \(n=2,4\), where their common
sizes are \(2\) and \(5\), respectively. When \(3\mid n\), the zero class
is optimal and the other two cyclic classes each have one fewer word. Among \emph{all} codes
the zero class is then the unique optimum unless \(n=6\)
\cite[Theorem~2.1 and Figure~2]{Machacek2021}. At \(n=6\) the three class
sizes are \(10,9,9\); there are exactly two maximum codes, and the second maximum code is not a cyclic syndrome class.

\paragraph{Four symbols.}
The alphabet may be identified either with \(\mathbb Z_4\) or with the
additive group of \(\mathbb F_4\). In either group a missing symbol is again a
syndrome difference. The known graph optimum is
\[
 S_4(n,1)=
 \begin{cases}
 (k+1)(k+2)(2k+3)/6,&n=2k+1,\\
 \binom{k+3}{3}+\binom{k+1}{3},&n=2k,
 \end{cases}
\]
by~\cite[Proposition~5.6 and Remark~5.7(ii)]{GeramitaGregoryRoberts1986}.
We show that the field labels realize both values. For odd \(n\), translating all labels
by \(g\in\mathbb F_4\) shifts the syndrome by \(ng=g\), so all four classes have
size \(V_{n,4}/4\), the displayed optimum. For even \(n=2k\), use labels
\(0,1,\alpha,1+\alpha\). The zero-syndrome condition says that the last
three multiplicities have the same parity; the even length \(n\) forces the
first one to have that parity as well. The all-even and all-odd counts give
the two binomial coefficients above. Machacek~\cite[Theorem~3.1]{Machacek2021}
proves uniqueness for even \(n\); for odd \(n\) the four field classes give
distinct optima.

The cyclic group also attains \(V_{n,4}/4\) at odd \(n\). For \(n=2k\),
Proposition~\ref{prop:sum-modulo-balance} gives its largest class size as
\(V_{n,4}/4+(k+1)/4+\mathbf 1_{2\mid k}/2\): the order-two character favors
even residues, and the two characters of order four contribute only if \(k\)
is even. The maximizing residues are \(0,2\) for odd \(k\), and \(0\) for
even \(k\). Subtraction from the field optimum gives exactly
\(\lceil n/4\rceil\) additional codewords for the field at every positive
even \(n\); the ratio of the two sizes is \(1+O(n^{-2})\).
These constructions differ in turn from the radix construction, whose
four labels for \(t=1\) lie in \(\mathbb Z_5\).

\paragraph{Prime alphabets.}
For \(q=5,7\) we use the third author's graph results
\cite[Theorems~6.6 and~7.5]{ZabokritskiyWeightedLocalCovers2026}.
The proofs in that preprint combine weighted covers with computer-verified
finite cases; we quote the results without repeating those checks here.
The largest cyclic class has size \(\lceil V_{n,q}/q\rceil\). It is optimal
except at the lengths in Table~\ref{tab:prime-exceptions}. For instance,
with seven symbols and length six it has \(132\) words, one short of the
optimum. The optima at length four also follow from the weight-four
formula of Chen, Ma and Zhang~\cite[Theorem~III.2]{ChenMaZhang2021}.

\begin{table}[htbp]
\centering
\caption{Lengths at which the largest cyclic class is smaller than an
optimal one-deletion code, for five and seven symbols.}
\label{tab:prime-exceptions}
\begin{tabular}{@{}rrrr@{}}
\toprule
Alphabet \(q\)&Length \(n\)&Largest cyclic class&Optimum \(S_q(n,1)\)\\
\midrule
5 & 2 & 3 & 5\\
5 & 4 & 14 & 16\\
7 & 2 & 4 & 7\\
7 & 3 & 12 & 14\\
7 & 4 & 30 & 35\\
7 & 6 & 132 & 133\\
\bottomrule
\end{tabular}
\end{table}

\begin{conjecture}[Eventual optimality]\label{conj:prime-eventual-single}
For each prime \(q\), \(S_q(n,1)=\lceil V_{n,q}/q\rceil\) for all sufficiently
large \(n\).
\end{conjecture}
The preceding results settle the conjecture for \(q=2,3,5,7\), with
respective thresholds \(1,5,5,7\), the least lengths from which equality holds. For every prime \(q\ge7\), lengths
\(2,3,4\) are all exceptions. At length two the \(q\) constant words
exceed the cyclic class of size \((q+1)/2\). At length three the exact
triple-packing formula above gives an excess of \((q-1)/3\) for
\(q\equiv1\pmod6\), and \((q-5)/3\) for \(q\equiv5\pmod6\).
At length four take all constant words, all words of type \(2+2\), and
all four-element sets of sum zero modulo \(q\). The three families have
disjoint output types. An output of a word of type \(2+2\) identifies both
of its symbols, while a squarefree triple determines its missing fourth symbol
by its sum. Translating the four-element sets cycles through all sums,
so there are \(\binom q4/q\) of sum zero. This code exceeds the cyclic
class of size \(\binom{q+3}{4}/q\) by \((q-1)/2\). These exceptions force a threshold of at least
five.

\section{Conclusions}\label{sec:conc}
For \(q,t\ge2\), the map \(b\mapsto tb+1\) permutes Nathanson's geometric labels modulo \(M\). We used this symmetry to determine the sharp exponent for class-size deviations. For prime \(q\), it also gives equality of all class sizes at the last \(q-1\) blocklength residues modulo \(M\). For \(q\ge4\) this modulus is smaller than the radix modulus; for \(q=2,3\) the two constructions coincide up to an alphabet permutation. For \(q\ge4\) the radix classes are smaller on average, but for even \(q\ge4\) they have a smaller deviation exponent about their own average. For four symbols the same radix labels still correct up to \(t\) deletions with a modulus smaller than both original moduli, while their deviations remain bounded. For fixed \(q\ge6\), the six initial labels give bounded deviations and, as \(t\) grows, a prime modulus asymptotic to \(t^{q-1}/3\). We count the arithmetic operations that a decoder performs, once the syndrome of the received word is available, to recover the deletions from the deletion count and the syndrome difference; this count excludes reading the received word and is not a bit-complexity measure. For fixed \(q\ge6\), what is the infimum of the leading coefficient of the modulus, as \(t\to\infty\), over cyclic syndrome constructions with bounded deviations whose decoders use \(O_q(1)\) arithmetic operations and \(O_q(1)\) stored integers, both uniformly in \(t\)?

When \(2\le k\le n\) symbols survive, Theorem~\ref{thm:bounded-multiplicity} reduces the problem to finding a maximum constant-weight code in the \(\ell_1\)-metric with entries below \(k\), weight \(n\) and minimum distance \(2(n-k+1)\), to which the \(q\) constant words are added. At blocklength four and for \(q\ge4\), we proved that the minimum total weight of a fractional cover, rounded down, equals \(S_q(4,1)\) precisely when a Steiner quadruple system of order \(q\) exists. For a single deletion, the exact sizes of the syndrome classes for the labels \(0,\ldots,q-1\) in \(\mathbb Z_q\) are known, but the largest class can be smaller than an optimal code. Conjecture~\ref{conj:prime-eventual-single} states that for each prime \(q\), the largest such class is optimal at all sufficiently large blocklengths. The conjecture holds for \(q=2,3,5,7\); the cases \(q=5,7\) rest on the third author's results~\cite{ZabokritskiyWeightedLocalCovers2026}, whose proofs include computer-verified finite cases.

\section*{Acknowledgments}
We thank the anonymous referees for their observations on the cyclic labels,
digit decoding, the four-symbol balance estimate, lattice scaling, and the
reduction for a fixed number of survivors.

\setlength{\bibsep}{4pt plus 1pt minus 1pt}


\begin{thebibliography}{99}
\interlinepenalty=10000

\bibitem{BeckSamWoods2008}
M.~Beck, S.~V.~Sam, and K.~M.~Woods,
Maximal periods of (Ehrhart) quasi-polynomials,
\emph{Journal of Combinatorial Theory, Series A} 115 (2008), 517--525,
\href{https://doi.org/10.1016/j.jcta.2007.05.009}{doi:10.1016/j.jcta.2007.05.009}.

\bibitem{BoseChowla1960}
R.~C.~Bose and S.~Chowla,
``Theorems in the additive theory of numbers,''
\emph{Commentarii Mathematici Helvetici},
vol.~37, pp.~141--147, 1962,
\href{https://doi.org/10.1007/BF02566968}{doi:10.1007/BF02566968}.

\bibitem{ChenMaZhang2021}
T.~Chen, Y.~Ma, and X.~Zhang,
``Optimal codes with small constant weight in \(\ell_1\)-metric,''
\emph{IEEE Transactions on Information Theory},
vol.~67, no.~7, pp.~4239--4254, July 2021,
\href{https://doi.org/10.1109/TIT.2021.3052191}{doi:10.1109/TIT.2021.3052191}.

\bibitem{CullinaKiyavash2016}
D.~Cullina and N.~Kiyavash,
Generalized sphere-packing bounds on the size of codes for combinatorial channels,
\emph{IEEE Transactions on Information Theory} 62 (2016), 4454--4465,
\href{https://doi.org/10.1109/TIT.2016.2565578}{doi:10.1109/TIT.2016.2565578}.
Preprint: \href{https://arxiv.org/abs/1405.1464v2}{arXiv:1405.1464v2}.

\bibitem{FazeliVardyYaakobi2015}
A.~Fazeli, A.~Vardy, and E.~Yaakobi,
Generalized sphere packing bound,
\emph{IEEE Transactions on Information Theory} 61 (2015), 2313--2334,
\href{https://doi.org/10.1109/TIT.2015.2413418}{doi:10.1109/TIT.2015.2413418}.

\bibitem{GeramitaGregoryRoberts1986}
A.~V.~Geramita, D.~Gregory, and L.~Roberts,
``Monomial ideals and points in projective space,''
\emph{Journal of Pure and Applied Algebra},
vol.~40, pp.~33--62, 1986,
\href{https://doi.org/10.1016/0022-4049(86)90029-0}{doi:10.1016/0022-4049(86)90029-0}.

\bibitem{Hanani1960}
H.~Hanani,
On quadruple systems,
\emph{Canadian Journal of Mathematics} \textbf{12} (1960), 145--157.
\href{https://doi.org/10.4153/CJM-1960-013-3}{doi:10.4153/CJM-1960-013-3}.

\bibitem{KiahPuleoMilenkovic2016}
H.~M.~Kiah, G.~J.~Puleo, and O.~Milenkovic,
``Codes for DNA sequence profiles,''
\emph{IEEE Transactions on Information Theory},
vol.~62, no.~6, pp.~3125--3146, June 2016,
\href{https://doi.org/10.1109/TIT.2016.2555321}{doi:10.1109/TIT.2016.2555321}.

\bibitem{KovacevicCross2022}
M.~Kova\v{c}evi\'c,
Lattice packings of cross-polytopes from Reed--Solomon codes and Sidon sets,
\emph{Bulletin of the London Mathematical Society} 54 (2022), 2372--2378,
\href{https://doi.org/10.1112/blms.12699}{doi:10.1112/blms.12699}.

\bibitem{KovacevicTan2017}
M.~Kova\v{c}evi\'c and V.~Y.~F.~Tan,
Improved bounds on Sidon sets via lattice packings of simplices,
\emph{SIAM Journal on Discrete Mathematics} 31 (2017), 2269--2278,
\href{https://doi.org/10.1137/16M1099182}{doi:10.1137/16M1099182}.

\bibitem{KovacevicTan2018}
M.~Kova\v{c}evi\'c and V.~Y.~F.~Tan,
``Codes in the space of multisets---coding for permutation channels with impairments,''
\emph{IEEE Transactions on Information Theory},
vol.~64, no.~7, pp.~5156--5169, July 2018,
\href{https://doi.org/10.1109/TIT.2017.2789292}{doi:10.1109/TIT.2017.2789292}.

\bibitem{KovacevicVukobratovic2013}
M.~Kova\v{c}evi\'c and D.~Vukobratovi\'c,
``Multiset codes for permutation channels,''
\emph{arXiv preprint}, \href{https://arxiv.org/abs/1301.7564}{arXiv:1301.7564}, 2013.

\bibitem{KovacevicVukobratovic2015}
M.~Kova\v{c}evi\'c and D.~Vukobratovi\'c,
``Perfect codes in the discrete simplex,''
\emph{Designs, Codes and Cryptography},
vol.~75, no.~1, pp.~81--95, Apr. 2015,
\href{https://doi.org/10.1007/s10623-013-9893-5}{doi:10.1007/s10623-013-9893-5}.

\bibitem{KreindelEssayagZabokritskiy2026}
A.~Kreindel, I.~B.~Essayag, and A.~L.~Zabokritskiy (Yohananov),
``Polynomial constructions and deletion-ball geometry for multiset deletion codes,''
\emph{arXiv preprint},
\href{https://arxiv.org/abs/2603.18322v1}{arXiv:2603.18322v1}, 2026.

\bibitem{Machacek2021}
J.~Machacek,
``Unique maximum independent sets in graphs on monomials of a fixed degree,''
\emph{Procedia Computer Science},
vol.~195, pp.~289--297, 2021,
\href{https://doi.org/10.1016/j.procs.2021.11.036}{doi:10.1016/j.procs.2021.11.036}.

\bibitem{MilenkovicPan2024}
O.~Milenkovic and C.~Pan,
``DNA-based data storage systems: A review of implementations and code constructions,''
\emph{IEEE Transactions on Communications},
vol.~72, no.~7, pp.~3803--3828, July 2024,
\href{https://doi.org/10.1109/TCOMM.2024.3367748}{doi:10.1109/TCOMM.2024.3367748}.

\bibitem{MuratovicWang2015}
A.~Muratovi\'c-Ribi\'c and Q.~Wang,
The multisubset sum problem for finite abelian groups,
\emph{Ars Mathematica Contemporanea} 8 (2015), 417--423,
\href{https://doi.org/10.26493/1855-3974.566.0da}{doi:10.26493/1855-3974.566.0da}.

\bibitem{Nathanson2025}
M.~B.~Nathanson,
The third positive element in the greedy \(B_h\)-set,
\emph{Palestine Journal of Mathematics} 14 (2025), 213--216,
\href{https://arxiv.org/abs/2310.14426}{arXiv:2310.14426}.

\bibitem{NathansonOBryant2024}
M.~B.~Nathanson and K.~O'Bryant,
The fourth positive element in the greedy \(B_h\)-set,
\emph{Journal of Integer Sequences} 27 (2024), Article 24.7.3.
\href{https://cs.uwaterloo.ca/journals/JIS/VOL27/OBryant2/obryant4.pdf}{Journal version}.

\bibitem{OBryant2025}
K.~O'Bryant, Bounds for greedy \(B_h\)-sets,
\emph{Electronic Journal of Combinatorics} 32 (2025), no.~2, Paper 2.39.
\href{https://doi.org/10.37236/12977}{doi:10.37236/12977}.

\bibitem{Panyushev2011}
D.~I.~Panyushev,
Fredman's reciprocity, invariants of abelian groups, and the permanent of the Cayley table,
\emph{Journal of Algebraic Combinatorics} 33 (2011), 111--125,
\href{https://doi.org/10.1007/s10801-010-0236-6}{doi:10.1007/s10801-010-0236-6}.

\bibitem{Rodl1985}
V.~R\"odl,
``On a packing and covering problem,''
\emph{European Journal of Combinatorics},
vol.~6, no.~1, pp.~69--78, 1985,
\href{https://doi.org/10.1016/S0195-6698(85)80023-8}{doi:10.1016/S0195-6698(85)80023-8}.

\bibitem{Varshamov1973}
R.~R.~Varshamov,
``A class of codes for asymmetric channels and a problem from the additive theory of numbers,''
\emph{IEEE Transactions on Information Theory},
vol.~19, no.~1, pp.~92--95, Jan. 1973,
\href{https://doi.org/10.1109/TIT.1973.1054954}{doi:10.1109/TIT.1973.1054954}.

\bibitem{WeiZhang2023}
X.~Wei and X.~Zhang,
``New results on ternary constant weight codes under \(\ell_1\)-metric,''
\emph{Scientia Sinica Mathematica},
vol.~53, no.~2, pp.~325--338, 2023, in Chinese,
\href{https://doi.org/10.1360/SSM-2022-0051}{doi:10.1360/SSM-2022-0051}.

\bibitem{XiaoZhou2024}
A.~Xiao and Y.~Zhou,
``On the packing density of Lee spheres,''
\emph{Designs, Codes and Cryptography},
vol.~92, no.~10, pp.~2705--2729, Oct. 2024,
\href{https://doi.org/10.1007/s10623-024-01410-0}{doi:10.1007/s10623-024-01410-0}.

\bibitem{ZabokritskiyWeightedLocalCovers2026}
A.~L.~Zabokritskiy (Yohananov),
``Independent sets in multiset profile graphs via weighted local covers,''
\emph{arXiv preprint}, \href{https://arxiv.org/abs/2607.13733v3}{arXiv:2607.13733v3}, 4 August 2026.
\end{thebibliography}
\end{document}